\documentclass{article}

\usepackage{graphicx}
\usepackage{amsmath,amssymb,amsthm}
\usepackage{natbib}
\usepackage{tikz}
\usetikzlibrary{arrows.meta}
\usepackage{caption}
\usepackage{booktabs}

\input{glyphtounicode}
\usepackage[T1]{fontenc}
\usepackage{sourceserifpro}
\usepackage{sourcecodepro}
\usepackage[hyphens]{url}

\usepackage{geometry}
\usepackage[onehalfspacing]{setspace}
\usepackage{microtype}
\AtBeginDocument{\frenchspacing}

\usepackage[x11names]{xcolor}
\usepackage{fancyhdr}
\fancypagestyle{plain}{\fancyhf{}}

\usepackage{titling}
\pretitle{ \begin{center}\bfseries\Huge}
\posttitle{\par \end{center} \vskip 0.5em}

\newcommand{\monthyear}{%
  \ifcase\month\or
    January\or February\or March\or April\or May\or June\or
    July\or August\or September\or October\or November\or December%
  \fi\space \number\year
}

\date{\monthyear}

\usepackage{etoolbox}
\AtBeginEnvironment{titlepage}{%
  \pagestyle{empty}%
  \setlength{\parindent}{0pt}%
}

\usepackage{titlesec}
\titleformat*{\section}{\centering\large\bfseries}
\titleformat*{\subsection}{\bfseries}
\titleformat{\paragraph}[runin]{\itshape}{}{}{}[.]
\titlelabel{\thetitle.\quad}

\usepackage[
  pdftex,
  hidelinks,
  hypertexnames=false,
  pdfpagemode=UseNone,
  pdfdisplaydoctitle=true
]{hyperref}

\newtheorem{proposition}{Proposition}

\newtheorem{assumption}{Assumption}

\hypersetup{pdftitle={Market Power and Platform Design in Decentralized Electricity Trading}}

\title{Market Power and Platform Design in Decentralized Electricity Trading}
\author{Nicolas Eschenbaum\thanks{Swiss Economics, Zürich, Switzerland. \href{mailto:nicolas.eschenbaum@swiss-economics.ch}{nicolas.eschenbaum@swiss-economics.ch}. This work was supported by the Swiss Federal Office of Energy (SFOE) under grant number SI/502525.} \and Nicolas Greber\thanks{University of Zurich and Swiss Economics, Zürich, Switzerland. \href{mailto:nicolasjoel.greber@uzh.ch}{nicolasjoel.greber@uzh.ch}.}}

\begin{document}

\begin{titlepage}
\maketitle
\thispagestyle{plain}

\begin{abstract}
This paper studies how platform design shapes strategic behavior in decentralized electricity trading. We develop a finite-horizon dynamic game in which photovoltaic- and battery-equipped players (`prosumers') trade on a platform that maps aggregate imports and exports into internal buy and sell prices. We establish existence of a perfect conditional $\varepsilon$-equilibrium and characterize a Cournot-like market-power mechanism in an observable-types benchmark of the game: because the producer price is decreasing in aggregate exports, strategic prosumers withhold supply and underutilize storage relative to the price-taking benchmark. To quantify these effects, we use a multi-agent computational framework that exploits the differentiable structure of the platform’s clearing rule to compare planner, price-taking, and strategic outcomes under alternative pricing mechanisms. In our baseline calibration, strategic play raises grid settlement cost by about 6 percent relative to price-taking. The magnitude of the distortion depends strongly on platform design: some designs can largely eliminate strategic incentives, while increased competition in storage ownership sharply reduces withholding, with most of the distortion disappearing once storage is split across (more than) three owners. We also find that information disclosure can improve competitive coordination but also increase the market power effects. Despite these distortions, the platform remains highly valuable overall, reducing a passive consumer's annual electricity bill by roughly 40 percent relative to exclusive grid settlement, with strategic behaviour clawing back only about 8 percent of that saving. The results show that pricing rules, information disclosure, and ownership structure determine how much of the gains from decentralized electricity trading are realized.
\end{abstract}

\medskip


\medskip


\end{titlepage}

\clearpage


\section{Introduction}\label{sec:intro}

Digital platforms increasingly intermediate transactions in markets that were historically organized through centralized pricing or bilateral contracting. Electricity is a particularly interesting setting, because participants trade a homogeneous good in repeated interactions, the outside option is regulated by retail and feed-in tariffs that are fixed by regulators, and users may control decentralized productive assets such as photovoltaic generation and battery storage. In this environment, a platform is not merely a coordination device. By mapping aggregate buying and selling quantities into internal prices, it becomes a market institution that shapes strategic behavior.

This paper studies market power and platform design in decentralized electricity trading. We consider a platform that clears local trade inside a community of prosumers and settles residual imbalances with the external grid. The platform precommits to a pricing rule that maps aggregate imports and exports into internal buy and sell prices, subject to a tariff corridor given by the regulated import and export prices. Prosumers choose when to charge or discharge their battery, whether to self-consume or sell locally, and when to buy from the platform rather than the grid. Because individual decisions affect aggregate market conditions, they also affect the prices that prosumers themselves face. The resulting game therefore combines intertemporal arbitrage incentives with price effects.

To study this setting, we formulate prosumer trading as a finite-horizon dynamic game with batteries, photovoltaics (PV), platform pricing, and private information. Each prosumer privately observes her type---a sequence of electricity consumption and PV generation realizations---and her own battery history, so that the current state of charge is payoff-relevant but not publicly observed. This yields a dynamic game with continuous actions and expanding private histories, a setting in which standard existence results do not directly apply. We show that the game can be embedded in the framework of \citet{myerson2020epsilon}, which implies existence of a perfect conditional $\varepsilon$-equilibrium for every $\varepsilon>0$.

We then characterize equilibrium behavior and the market-power mechanism in a transparent benchmark of the game with publicly observable types and battery states. We show that a pure-strategy Markov perfect equilibrium exists under mild curvature conditions, and that net exporters withhold supply relative to the corresponding price-taking benchmark. The intuition is classic Cournot: when the producer price is decreasing in aggregate exports, an exporting prosumer values marginal sales at marginal revenue rather than price. The same logic has a dynamic analogue. If an additional unit stored today is expected to be sold later, then a strategic prosumer charges less than a price-taker because future sales are discounted by the effect of own exports on future prices. Decentralized storage therefore creates a prosumer-scale version of a familiar Cournot finding. Strategic players underinvest in capacity in equilibrium (i.e., underutilize storage in our setting) when discharging lowers the price they receive. As a result, an intuitive ordering of outcomes emerges: total grid settlement cost is lowest when all assets are controlled by a social planner, intermediate when prosumers are price-takers, and highest when prosumers play strategically. Since internal transfers between players and the platform cancel, grid cost is the appropriate measure of total welfare, and market power effects are thus costly for welfare.

We then apply a computational framework developed in \cite{EschenbaumGreberSzehr2026} to compare alternative platform rules and provide quantitative measurements of the outcome of the game. We model each prosumer by a neural-network policy and train agents in a differentiable multi-agent environment in which the platform's clearing rule is known. Because prices are a differentiable function of aggregate trades, gradients can be propagated through the clearing routine and through the full multi-period trajectory. This allows us to compare, within a common architecture, three counterfactuals: a centralized planner, price-taking prosumers, and strategic prosumers. The unilateral-deviation regret is consistently very low, confirming that the strategic outcome is an approximate Nash equilibrium, making us confident that agents converged to a perfect conditional $\varepsilon$-equilibrium. We use this framework to study pricing mechanisms commonly used in local electricity markets and to quantify how competition, heterogeneity, and information design affect equilibrium outcomes.

Our quantitative results confirm the theoretical ordering. The social planner achieves the lowest grid cost, the price-taking benchmark lies close behind, and strategic play raises cost by roughly 6\%. Market power manifests by prosumers withholding evening discharge of their batteries. As competition among storage ownership increases, it erodes market power. Moving from a single storage owner to three eliminates most of the strategic distortion, and the welfare loss is predominantly deadweight rather than rent transfer. Across pricing mechanisms however, the effect of market power incentives vary by an order of magnitude. A mid-market rate mechanism virtually removes strategic incentives, while also removing incentives for pure arbitrage. The linear mechanism permits individual withholding but limits aggregate efficiency loss. Allowing agents to observe past clearing prices improves competitive outcomes but amplifies strategic distortions, more than doubling the welfare premium from market power.

Overall, however, the platform remains highly valuable. A passive consumer saves roughly 40\% on annual electricity costs by trading on the platform. Strategic behaviour by battery owners erodes about 8\% of that saving. The policy implication is therefore not that platforms should be avoided, but that their design, in particular the pricing rule, information disclosure, and the competitive structure of storage ownership, materially affects who captures the surplus from decentralized trade.

The paper contributes to three strands of the literature. The first studies storage and market power in electricity markets. Existing work shows that storage can smooth prices and improve welfare under competition, but that concentrated ownership can also lead to strategic underutilization of storage capacity \citep{andrescerezo2023storing,schill2011strategic,sioshansi2014storage,garcia2001strategic,bushnell2003looking,williams2022storage}. Related work emphasizes that market design and ownership structure are central for determining whether decentralized trade improves welfare \citep{baake2023local}. Our contribution is to bring these ideas to prosumer-scale storage on trading platforms and to make the platform pricing rule itself part of the strategic environment.

The second strand of literature studies peer-to-peer and community energy markets, often from the perspective of market architecture, clearing rules, and network constraints \citep{parag2016prosumer,mengelkamp2018brooklyn,morstyn2018federated,le2020peer,zhang2020cournot,etesami2018stochastic,irena2020p2p,tushar2021p2p_review}. Related contributions study battery-enabled peer-to-peer sharing, local pricing, and prosumer incentives such as flexibility provision and cost recovery \citep{he2021energypawn,hoseinpour2024costrecovery}. Relative to that literature, we focus on strategic price impact in a dynamic game and on how mechanism design and information design shape the exercise of market power.

The third strand of literature uses machine learning to analyze strategic interaction and evaluate economic designs. \citet{calvano2020artificial} show that reinforcement-learning algorithms in repeated pricing games can learn supracompetitive outcomes, \citet{zheng2022aieconomist} develop a multi-agent learning framework for tax-policy design, and \citet{curry2022finding} use deep multi-agent reinforcement learning to compute approximate equilibria in microfounded general-equilibrium environments. Closer to our application, \citet{kastius2022dynamic} study reinforcement learning in competitive dynamic pricing. Our computational contribution differs from model-free approaches because we exploit the known, differentiable structure of the platform's clearing rule and backpropagate through the full trading environment, enabling a controlled comparison between strategic and price-taking behavior. More broadly, the paper connects to the industrial organization of electricity markets \citep{fabra2021energy,fabra2023market,wolfram1999measuring,hortacsu2008understanding}, bringing themes of strategic bidding and market power measurement to the emerging setting of prosumer trading platforms.

The remainder of the paper is structured as follows. Section~\ref{sec:fullmodel} introduces the model. Section~\ref{sec:existence-private} establishes equilibrium existence in the private-information game. Section~\ref{sec:benchmark} analyzes market power in the observable-types benchmark. Section~\ref{sec:computation} presents the computational framework, data, and pricing mechanisms. Section~\ref{sec:results} reports the quantitative results. Section~\ref{sec:conclusion} concludes.

\section{A model of prosumer trading}\label{sec:fullmodel}

This section develops the model of prosumer trading that we are studying in this paper. The players (i.e.\ prosumers) may have access to a battery to engage in intertemporal arbitrage and self-consumption of their PV-generated electricity. As the battery state is private information, but payoff-relevant for all players, the game features private information with infinite action spaces, implying standard existence theorems do not apply and making equilibrium characterization complex. Therefore, we first establish equilibrium existence in the full private-information game in \autoref{sec:existence-private} by mapping it into the framework of \cite{myerson2020epsilon}, implying that a perfect conditional $\varepsilon$-equilibrium exists for every $\varepsilon>0$. In \autoref{sec:benchmark} we then study a simplified benchmark setting in which players' types and battery states are publicly observed. This allows us to characterize the static and dynamic market power effects that arise in a transparent complete-information setting. In what follows, we will refer to these two settings as the `private-types' and `observable-types' settings respectively.

\subsection{Setup}\label{subsec:setup}

Time is $t\in\{1,\dots,T\}$ and players are $i\in I:=\{1,\dots,N\}$, where $N<\infty$ and $T<\infty$. At date $0$, nature draws a type profile $\theta=(\theta_i)_{i\in I}\in\Theta:=\prod_{i\in I}\Theta_i$, where each $\Theta_i$ is finite, according to a common prior $\pi\in\Delta(\Theta)$. Player $i$ privately observes $\theta_i$ at the start of period $1$. A type $\theta_i$ specifies a $T$-period sequence of inelastic demand and supply,
\[
\theta_i = \bigl((d_i^t(\theta_i),s_i^t(\theta_i))\bigr)_{t=1}^T,
\qquad
\ell_i^t(\theta_i):=s_i^t(\theta_i)-d_i^t(\theta_i).
\]
We interpret $\ell_i^t(\theta_i)$ as the exogenous physical surplus (deficit) in period $t$.\footnote{Assuming perfectly inelastic demand for electricity is standard in the literature. However, because players in our setting can adjust their battery flow (which affects their total demand), this assumption could be relaxed as long as players per-period demand remains bounded.}

Player $i$ has a battery with state of charge $S_{it}\in[0,\bar S_i]$, with $S_{i1}=0$ (common knowledge). In period $t$, player $i$ chooses a battery flow $b_{it}$, where a positive value denotes charging and a negative value discharging, which updates the battery via
\[
S_{i,t+1}=S_{it}+b_{it}.
\]
Feasibility requires power bounds $b_{it}\in[\underline b_i,\bar b_i]$ and capacity bounds
$0\le S_{it}+b_{it}\le \bar S_i$. Hence the feasible set given $S_{it}$ is
\[
B_i(S_{it})
:=
\bigl[\max\{\underline b_i,-S_{it}\},\,\min\{\bar b_i,\bar S_i-S_{it}\}\bigr].
\]
We assume parameters are such that $B_i(S)$ is nonempty for all $S\in[0,\bar S_i]$.
Given $(S_{it},\theta_i)$, the induced net import (net trade) is
\[
x_{it} := d_i^t(\theta_i)+b_{it}-s_i^t(\theta_i)=b_{it}-\ell_i^t(\theta_i).
\]
Define $x^+:=\max\{x,0\}$ and $x^-:=\max\{-x,0\}$ so that $x=x^+-x^-$. Imports and exports are then given by
\[
m_{it}:=x_{it}^+,\qquad q_{it}:=x_{it}^-.
\]
Let $M_t:=\sum_{i\in I} m_{it}$, $Q_t:=\sum_{i\in I} q_{it}$, and $Y_t:=\sum_{i\in I} x_{it}=M_t-Q_t$.

\subsection{Prices, platform settlement, and payoffs}\label{subsec:prices-payoffs}

External grid prices are given by $(P^E,P^I)$ for exports/imports respectively and are constant throughout the game. The platform precommits to internal price schedules $p_c(M,Q)$ (consumer price) and $p_p(M,Q)$ (producer price), which are time-invariant and common knowledge.
Given bounded battery flows and finite $\Theta$, for each $t$ there exists a compact set $K_t\subset\mathbb R_+^2$ such that $(M_t,Q_t)\in K_t$ for all feasible action profiles and all $\theta\in\Theta$. Let $K:=\bigcup_{t=1}^T K_t$. We impose the following assumption on the price schedules on the attainable set $K$.

\begin{assumption}[Prices on the attainable set]\label{ass:prices1}
The functions $p_c$ and $p_p$ are continuous on $K$ and satisfy for all $(M,Q)\in K$,
\[
0 \le P^E \le p_p(M,Q) \le p_c(M,Q) \le P^I < \infty.
\]
\end{assumption}

\noindent
Assumption~1 places weakly positive bounds on the platform prices, which are typically given by the relevant feed-in tariff (i.e., the export price) and energy tariff (i.e., the import price). We can then define player $i$'s period-$t$ payment to be
\[
c_{it} := p_c(M_t,Q_t)\,m_{it} - p_p(M_t,Q_t)\,q_{it},
\]
so player $i$'s per-period payoff is $u_{it}:=-c_{it}$, and total payoff is $u_i:=\sum_{t=1}^T u_{it}$. The platform settles any net imbalance $Y_t$ with the external grid at cost
\[
C_t(Y_t):=P^I\,Y_t^+ - P^E\,Y_t^-.
\]
Platform profit in period $t$ is given by
\[
\Pi_t := p_c(M_t,Q_t)M_t - p_p(M_t,Q_t)Q_t - C_t(Y_t).
\]
Thus, internal transfers cancel and total (players and platform) period-$t$ surplus equals $-C_t(Y_t)$.

\subsection{Information, timing, strategies, and beliefs}\label{subsec:info-timing}

In each period $t\in\{1,\dots,T\}$, 
(i) players simultaneously choose $b_{it}\in B_i(S_{it})$;
(ii) the platform computes $(M_t,Q_t)$ from the action profile and $\theta$, and (iii) payments $c_{it}$ are realized and the battery state updates via $S_{i,t+1}=S_{it}+b_{it}$.
Players do not observe payments or opponents' actions; player $i$'s only private information is her type $\theta_i$ (observed at $t=1$) and her own past battery flows.\footnote{Payments $c_{it}$ are deterministic functions of the full action profile and types. We exclude them from observed signals in the primitive description for expositional convenience. In \autoref{sec:existence-private} we show that allowing each player to observe her own realized payment (equivalently per-period payoff) is without loss via a purely formal reporting stage, while maintaining compatibility with the projected-signal structure required by \cite{myerson2020epsilon}.}
There is no payoff after period $T$, i.e.\ there is no salvage value.

Player $i$'s private history at the start of period $t$ is
\[
h_{it}:=\bigl(\theta_i,b_{i1},\dots,b_{i,t-1}\bigr),
\qquad h_{i1}:=\theta_i.
\]
The battery state $S_{it}=\sum_{\tau=1}^{t-1}b_{i\tau}$ (with $S_{i1}=0$) is computed from $h_{it}$, and feasibility requires $b_{it}\in B_i(S_{it}(h_{it}))$.
A (behavioral) strategy for player $i$ in period $t$ is a stochastic kernel
\[
\sigma_{it}(\cdot\mid h_{it})\in \Delta(B_i(S_{it})).
\]
A strategy profile is $\sigma=(\sigma_i)_{i\in I}$ with $\sigma_i=(\sigma_{it})_{t=1}^T$. A belief system is a collection $\mu=(\mu_{it})_{i\in I,\,t\in\{1,\dots,T\}}$ where each
\[
\mu_{it}(\cdot\mid h_{it})\in
\Delta \Bigl(\Theta_{-i}\times \prod_{j\neq i}[0,\bar S_j]\Bigr)
\]
represents player $i$'s posterior over $(\theta_{-i},S_{-i,t})$ at information set $h_{it}$. Beliefs are consistent with $(\pi,\sigma)$ if, for each player $i$ and period $t$, $\mu_{it}(\cdot\mid h_{it})$ coincides with a regular conditional distribution of $(\theta_{-i},S_{-i,t})$ given $h_{it}$ under the probability measure induced by $(\pi,\sigma)$.

Given $(\sigma,\mu)$, player $i$'s continuation value can be defined recursively by
\[
V_{it}(h_{it};\sigma,\mu)
:=
\mathbb E_{\pi,\sigma,\mu} \left[\sum_{\tau=t}^T u_{i\tau}\,\Big|\, h_{it}\right],
\qquad V_{i,T+1}\equiv 0.
\]
The equilibrium concept for the private-information game is introduced in \autoref{sec:existence-private}. For the observable-types benchmark we use Markov perfect equilibrium as formally defined in \autoref{sec:benchmark}.

\section{Equilibrium analysis}

\subsection{Equilibrium existence with private types}\label{sec:existence-private}

We establish equilibrium existence in the full private-information game of \autoref{sec:fullmodel} by mapping it into the framework of \cite{myerson2020epsilon}. Their Theorem~9.3 guarantees existence of a perfect conditional $\varepsilon$-equilibrium for every $\varepsilon>0$ in any regular projective game. We verify that our game satisfies their conditions R.1--R.5.

\begin{proposition}[Private Types]\label{prop:private-existence}
For every $\varepsilon>0$, the private-information game possesses a perfect conditional $\varepsilon$-equilibrium [Myerson \& Reny, 2020, Theorem~9.3]. Moreover, this existence result is unchanged (without loss) if each player also observes her realized per-period payment $c_{it}$ (equivalently $u_{it}=-c_{it}$) at the end of each period.
\end{proposition}

Proposition~\ref{prop:private-existence} establishes that an approximate equilibrium exists in the full private-information game for any $\varepsilon>0$. This result is non-trivial because the game combines private information with continuous actions and an expanding private history: each player's battery state $S_{it}$ is private information (it is determined by past battery flows that only the owner observes) but it is payoff-relevant for all players through its effect on aggregate quantities and thus on prices. Standard existence theorems for finite Bayesian games do not directly apply in this setting.

The proof proceeds by mapping an equivalent representation of the game into the framework of \cite{myerson2020epsilon}. The key elements are: (i) representing battery feasibility via an intended action in a fixed compact interval together with an implemented flow given by projection onto $B_i(S_{it})$, which makes the available (intended) action set independent of the signal (Condition~R.1), and (ii) exploiting that each player's private signal is a projection of the outcome history onto her own coordinates (Condition~R.2). Continuity of payoffs on the attainable set (Assumption~\ref{ass:prices1}), compactness of action sets, and the finiteness of $\Theta$ establish the remaining regularity conditions. Allowing a player to observe her own per-period payment $c_{it}$ is without loss because it can be incorporated into the projected-signal structure by adding a purely formal reporting stage (as done in the proof), so the same existence conclusion continues to hold.

In a perfect conditional $\varepsilon$-equilibrium, no player can improve her expected continuation payoff by more than $\varepsilon$ at any information set, conditional on reaching that information set. The parameter $\varepsilon$ thus bounds the expected continuation gain from deviation in monetary units.
Having established equilibrium existence in the full private-information game, we now turn to the observable-types benchmark in \autoref{sec:benchmark} to characterize the market power mechanisms that drive strategic behavior.


\subsection{Market power effects with observable types}\label{sec:benchmark}

This section analyzes a benchmark setting in which the type profile $\theta$ is publicly observed before period $1$ and battery states are also publicly observed. All other primitives (battery constraints, price schedules, and payoffs) are unchanged. This allows us to provide a clean characterization of the static and dynamic incentives for players that are driven by market power.

In this setting, the public state at the start of period $t$ is
\[
\omega_t := (t,\theta,S_t),\qquad S_t:=(S_{it})_{i\in I}\in \prod_{i\in I}[0,\bar S_i],
\]
with $\omega_1=(1,\theta,0)$ where $0$ denotes the zero vector in $\prod_{i\in I}[0,\bar S_i]$. In each period $t\in\{1,\dots,T\}$, players observe $\omega_t$, choose $b_{it}\in B_i(S_{it})$ simultaneously, payoffs accrue, and the state updates via $S_{i,t+1}=S_{it}+b_{it}$.

A (mixed) Markov strategy for player $i$ is a collection $\beta_i=(\beta_{it})_{t=1}^T$ where, for each $t$,
\[
\beta_{it}(\cdot\mid \theta,S_t)\in \Delta\bigl(B_i(S_{it})\bigr),
\qquad (\theta,S_t)\in \Theta\times \prod_{j\in I}[0,\bar S_j].
\]
Let $\beta=(\beta_i)_{i\in I}$ be a Markov strategy profile. For each $t$ and public state $(\theta,S_t)$, $\beta$ induces a distribution over the action profile $b_t=(b_{it})_{i\in I}$ and thus over $(M_t,Q_t,Y_t)$ and the payoff vector $(u_{it})_{i\in I}$. The next-period state is $S_{t+1}=S_t+b_t$ componentwise.

Given $\beta$, define player $i$'s continuation value recursively by $W_{i,T+1}\equiv 0$ and, for $t\in\{1,\dots,T\}$,
\[
W_{it}(\theta,S_t;\beta)
:=
\mathbb E_{\beta} \left[u_{it}(b_t;\theta,S_t)+W_{i,t+1}(\theta,S_{t+1};\beta)\,\Big|\,\theta,S_t\right],
\]
where the expectation is taken over $b_t\sim \beta_t(\cdot\mid \theta,S_t):=\prod_{j\in I}\beta_{jt}(\cdot\mid \theta,S_t)$.

A Markov strategy profile $\beta^\ast$ is a \emph{Markov perfect equilibrium (MPE)} if for every $i\in I$, every $t\in\{1,\dots,T\}$, and every public state $(\theta,S_t)$, the distribution $\beta_{it}^\ast(\cdot\mid \theta,S_t)$ assigns probability $1$ to actions $b_{it}\in B_i(S_{it})$ that maximize
\[
\mathbb E_{\beta^\ast} \left[u_{it}\bigl((b_{it},b_{-i,t});\theta,S_t\bigr)
+W_{i,t+1}\bigl(\theta,S_t+(b_{it},b_{-i,t});\beta^\ast\bigr)
\,\Big|\,\theta,S_t,\ b_{it}\right]
\]
given $\beta_{-i,t}^\ast(\cdot\mid \theta,S_t)$.

We place the following shape and curvature assumptions on the price schedules.

\begin{assumption}[Monotonicity and curvature of price schedules]\label{ass:prices2}
The functions $p_c$ and $p_p$ are continuously differentiable and satisfy:
\begin{enumerate}
\item[(i)] The price schedule $p_c$ is weakly increasing in $M$ for each fixed $Q\ge 0$ and for every $Q\ge 0$ and every $M_{-i}\ge 0$, the total import-cost function
\[
m\,p_c(M_{-i}+m,\,Q)
\]
is convex in $m$.

\item[(ii)] The price schedule $p_p$ is weakly decreasing in $Q$ for each fixed $M\ge 0$ and for every $M\ge 0$ and every $Q_{-i}\ge 0$, the total export-revenue function
\[
q\,p_p(M,\,Q_{-i}+q)
\]
is  concave in $q$.
\item[(iii)] For the strict inequality in Proposition~\ref{prop:market-power}\emph{(ii)}, we additionally require that $p_p$ is strictly decreasing in $Q$ at the relevant arguments, i.e.\ $\partial_Q p_p(M,Q)<0$ whenever the solution is interior.
\end{enumerate}
\end{assumption}

\noindent
Assumption~\ref{ass:prices2} places intuitive shape assumptions on the price schedules. The consumption price increases in consumption (import) and the sell price decreases in supply (export). In addition, to guarantee existence of pure-strategy equilibria in the one-shot and continuation games in Proposition~\ref{prop:market-power}, the respective cost and revenue functions are assumed to be convex and concave, respectively. 

The MMR mechanism satisfies Assumption~\ref{ass:prices2} globally (see Appendix~\ref{app:verification}). The SDR and linear mechanisms satisfy the monotonicity conditions and the curvature conditions within each smooth pricing regime, but price clamping (for the linear mechanism) and regime-boundary effects (for SDR) can create kinks that violate global concavity or convexity. For these mechanisms, we verify in Appendix~\ref{app:verification} that the \emph{total} period-$t$ objective---stage payoff plus continuation value---remains concave in $b_{it}$ on the attainable set, which is the condition actually used in the proof of Proposition~\ref{prop:market-power}(i).

To analyze market power effects, we define two benchmarks for player $i$ at state $(\theta,S_t)$. Fix a period $t$ and a public state $(\theta,S_t)$ in the observable-types benchmark.
For any opponents' action profile $b_{-i,t}\in\prod_{j\neq i}B_j(S_{jt})$, let
$(M_{-i,t},Q_{-i,t})$ denote the induced aggregates and define player $i$'s (period-$t$) import/export quantities as
\[
m_{it}(b_{it}) := \bigl(b_{it}-\ell_i^t(\theta_i)\bigr)^+,
\qquad
q_{it}(b_{it}) := \bigl(\ell_i^t(\theta_i)-b_{it}\bigr)^+.
\]
Write $S_{t+1}=S_t+(b_{it},b_{-i,t})$.

Let $W^\ast_{i,t+1}(\theta,\cdot)$ denote player $i$'s equilibrium continuation value at $t+1$ induced by an MPE
$\beta^\ast$, and let $\mathcal L_{-i,t}^\ast(\cdot\mid \theta,S_t)$ denote the distribution of $b_{-i,t}$ induced by
$\beta^\ast_{-i,t}(\cdot\mid \theta,S_t)$.
In MPE, the best reply by player $i$ is given by $b_{it}\in B_i(S_{it})$ that maximizes
\[
\mathbb E_{b_{-i,t}\sim \mathcal L_{-i,t}^\ast}
\Bigl[
-\,p_c(M_t,Q_t)\,m_{it}(b_{it}) \;+\; p_p(M_t,Q_t)\,q_{it}(b_{it})
\;+\; W^\ast_{i,t+1}(\theta,S_{t+1})
\Bigr],
\]
where $(M_t,Q_t)$ depend on $(b_{it},b_{-i,t},\theta)$. Let $b_{it}^\ast$ denote the resulting battery action and write $q_{it}^{\ast}:=q_{it}(b_{it}^{\ast})$ for the corresponding export quantity.

The two benchmarks are defined as follows. First, we define \emph{price-taking} behavior as a benchmark in which the player optimizes individually but treats the internal price level as given when evaluating marginal deviations, i.e.\ the player ignores the effect of her own action on the price schedule. Formally, fix opponents' actions $b_{-i,t}$ and let player $i$'s objective be written as a function of $b_{it}\in B_i(S_{it})$. A price-taking action $b_{it}^{PT}$ is any action satisfying the KKT conditions of player $i$'s problem after removing the own-price-impact terms from the stationarity condition (the terms generated by $\partial_M p_c$ and $\partial_Q p_p$), while keeping the realized price levels $p_c(M_t,Q_t)$ and $p_p(M_t,Q_t)$ unchanged.

On the interior export branch ($q=\ell_i^t(\theta_i)-b_{it}>0$), with opponents fixed and $Q_t=Q_{-i,t}+q$, this means that the strategic stationarity condition
\[
p_p(M_{-i,t},Q_{-i,t}+q)+q\,\partial_Q p_p(M_{-i,t},Q_{-i,t}+q)+\partial_q \Psi_i(q)=0
\]
is replaced by the price-taking condition
\[
p_p(M_{-i,t},Q_{-i,t}+q)+\partial_q \Psi_i(q)=0,
\]
together with the corresponding boundary/complementarity conditions. Here $\Psi_i(q)$ denotes the continuation term expressed in the export variable $q$ (and $\Psi_i\equiv 0$ in the stage-game). We denote the associated export quantity by $q_{it}^{PT}$.

Second, consider that a social planner chooses $(b_{jt})_{j\in I}$ (and continuation actions) to minimize total grid cost $\sum_{\tau=t}^T C_\tau(Y_\tau)$ subject to all feasibility and battery constraints. Let $b_{it}^{SP}(\theta,S_t)$ denote player $i$'s period-$t$ component of a planner-optimal plan (if not unique, any selection) and  $q_{it}^{SP}:=q_{it}(b_{it}^{SP})$ the corresponding export quantity.


For the observable-types setting, we can then state the following result.

\begin{proposition}[Observable Types]\label{prop:market-power}
For the observable-types setting and under Assumptions~\ref{ass:prices1} and~\ref{ass:prices2}:
\begin{enumerate}
\item[\emph{(i)}]
A pure-strategy Nash equilibrium exists in the one-shot game. Moreover, for every $\varepsilon>0$ the dynamic observable-types game possesses a perfect conditional $\varepsilon$-equilibrium (by the same \cite{myerson2020epsilon} argument as in Proposition~\ref{prop:private-existence}). If the equilibrium continuation values $W_{it}^\ast(\theta,S_t)$ are concave in $S_{it}$ for each $i$, $t$, $\theta$, and $S_{-i,t}$, then a pure-strategy MPE exists in the dynamic game.
\item[\emph{(ii)}]
Fix a period $t$ and state $(\theta,S_t)$, and consider the associated one-shot game. For any player $i$ with an interior export solution ($q_{it}^{\ast}>0$),
\[
q_{it}^{\ast}\ \le\ q_{it}^{PT} = q_{it}^{SP}
\]
with strict inequality in the first comparison whenever $p_p(M,Q)$ is strictly decreasing in $Q$ at the relevant arguments.
\item[\emph{(iii)}]
Fix any $t<T$, player $i$, state $(\theta,S_t)$, and opponents' continuation path $(b_{-i,\tau})_{\tau=t}^T$.
If both period-$t$ solutions are interior import choices and the marginal stored unit is exported at some
$\tau>t$, then
\[
b_{it}^{\ast}\ \le\ b_{it}^{PT}.
\]
\item[\emph{(iv)}]
Fix any $t<T$ and state $(\theta,S_t)$. Let $(b_{i\tau}^{PT})_{i,\tau=t}^T$ denote the price-taking continuation from $(\theta,S_t)$, and let $(b_{i\tau}^{SP,t})_{i,\tau=t}^T$ be a planner-optimal continuation. If
\[
\sum_i b_{it}^{PT} \;<\; \sum_i b_{it}^{SP,t},
\]
then, if the additional period-$t$ charge is fully discharged through future exports, there exists some $\tau>t$ such that
\[
Q_\tau^{SP,t} \;>\; Q_\tau^{PT}.
\]

\end{enumerate}
\end{proposition}

Proposition~\ref{prop:market-power} clarifies the static and dynamic market-power effects that arise in the observable-types setting. Part~\emph{(i)} first establishes existence of a pure-strategy Nash equilibrium in the one-shot stage game. It also notes that the dynamic observable-types game admits a perfect conditional $\varepsilon$-equilibrium for every $\varepsilon>0$, and that under concavity of the equilibrium continuation values a pure-strategy MPE exists in the dynamic game. We make this concavity condition explicit because it is not directly implied by Assumptions~\ref{ass:prices1} and~\ref{ass:prices2}, but can be shown to hold for the mechanisms studied in this paper (see Appendix~\ref{app:verification}). Part~\emph{(ii)} then provides an intuitive ranking among exporting players in the stage game: prosumers will export weakly less than the corresponding price-taking benchmark, while the price-taking and social-planner export choices coincide. Part~\emph{(iii)} provides the dynamic analogue and similarly shows that a strategic prosumer charges weakly less than a price-taker because future sales are valued at marginal revenue rather than price. Finally, part~\emph{(iv)} clarifies the planner comparison from part~\emph{(ii)} for the dynamic game.  If a planner-optimal continuation charges more in period $t$ than the price-taking continuation and this additional charge is later discharged through exports, then the planner must export strictly more than the price-taking path in at least one future period. This implies that whenever prosumers undercharge relative to a price-taking player and thereby limit future exports, a stricter ranking of quantities supplied arises in the dynamic game than in the stage game. Note that the corresponding strategic behaviour does not arise among net importers for whom minimizing imports is always a best-response both in MPE and when being price-takers.

\begin{figure}[ht!]
\centering
\includegraphics[width=0.45\textwidth]{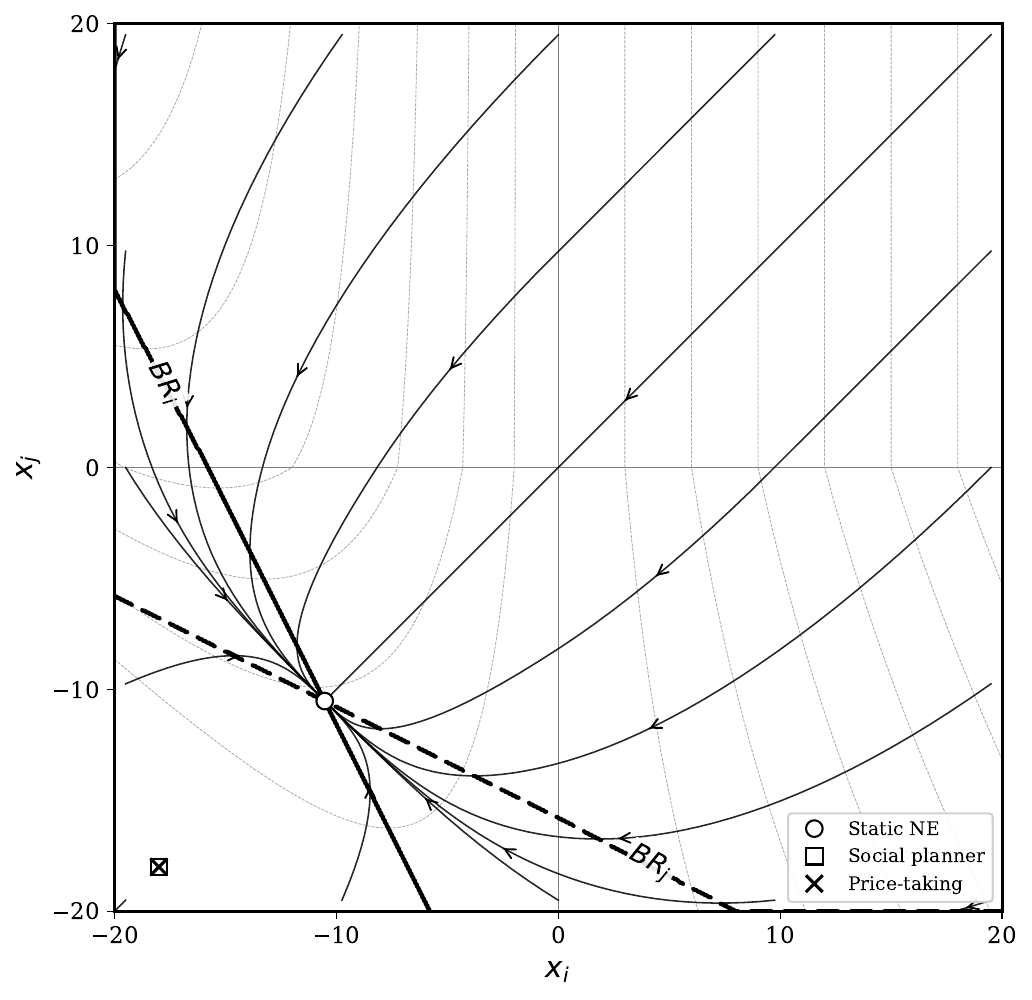}
\caption{Best-response curves and convergence to Nash equilibrium in a two-prosumer stage game with net exporters.}
\label{fig:intuition}
\end{figure}

Figure~\ref{fig:intuition} illustrates the strategic interaction for the simplest case: a single-period game  with two symmetric prosumers and fixed exogenous demand to ensure the price reacts to players' actions.\footnote{If both players are net exporters, there is no demand on the platform. The price mechanisms we consider will then default to the lower bound, i.e. the constant export price $P^E$. For the illustration of the strategic interaction we therefore add a positive exogenous demand.} In the illustrated case, both players have positive surplus ($\ell_i^t>0$) and hence the equilibrium lies in the both-export region ($x_i,x_j<0$), where the game becomes a Cournot duopoly against the exogenous demand. The best-response curves $BR_i$ (solid) and $BR_j$ (dashed) are downward-sloping---quantities are strategic substitutes---and intersect at the symmetric Nash equilibrium $x_i^\ast=x_j^\ast=-D_{\mathrm{ext}}/3$. Each player withholds exports relative to the price-taking optimum because she internalizes the effect of her supply on the market-clearing price; the resulting quantity is strictly less than the price-taking export. The social planner solution coincides with the price-taking optimum, illustrating the ranking developed in Proposition~\ref{prop:market-power} of $q_{it}^{SP} = q_{it}^{PT} \geq q_{it}^{\ast}$.

The gradient-play trajectories (thin curves with arrows) show the adjustment dynamics under simultaneous best-response improvement: from any initial action profile, the system converges to the Nash equilibrium. In the upper right quadrant, both players are net importers and the best-response for a net importer in the stage game is to minimize their imports. Thus the trajectory is a diagonal line. This can also be seen in the dashed indifference curves shown for player $i$ which are vertical in the top right quadrant: the best-response is to move horizontally towards $x_i=0$. In the top left and bottom right quadrants in turn, one player is a net exporter and the other player a net importer. In this case, the net importer continues to want to minimize imports, while the net exporter faces a convex indifference curve and engages in strategic withholding. Because of this strategic incentive, the curvature of player~$i$'s indifference curves reflects the concavity of the export-revenue function $q\,p_p(M,Q_{-i}+q)$ from Assumption~\ref{ass:prices2}(ii). In the dynamic game, these stage game incentives interact with the intertemporal incentives. First, a player who anticipates future market power effects may strategically undercharge her battery today in order to limit her exports tomorrow. Second, a player may charge in order to sell at a higher price in the future and engage in arbitrage. This incentive can be greater than the market power effects, depending on the realizations of types---including the expectations of future realizations---and the price schedule design.


\begin{figure}[ht!]
\centering
\includegraphics[width=1\textwidth]{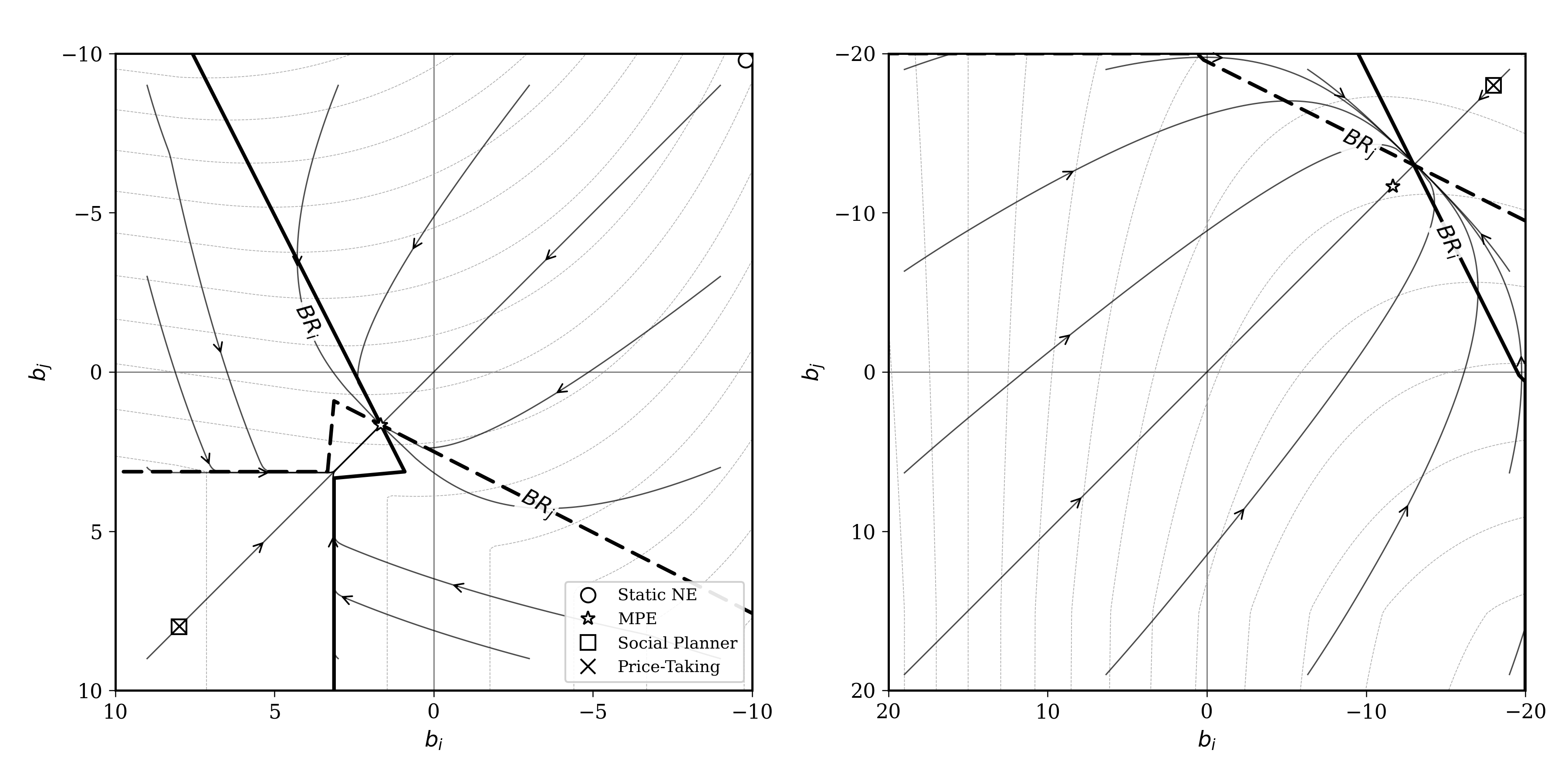}
\caption{Best-response curves and convergence to MPE in a two-prosumer, two-period game with net exporters.}
\label{fig:twostage-intuition}
\end{figure}

Figure~\ref{fig:twostage-intuition} illustrates this for the simplest case of two symmetric players and two periods with exogenous demand in both periods and a positive battery state in the first period. The stage game equilibrium in the first period (left panel) sees both players discharging their battery to the maximum. In MPE however, both players instead \emph{charge} their battery---incurring extra cost---in order to discharge more in the second period (right panel) when prices are higher, i.e. to engage in intertemporal arbitrage. But players charge less and subsequently also discharge less than either a price-taking player or the social planner would, implying they are withholding capacity because they correctly account for the price effect their discharging will have in the second period. As before, the gradient-play trajectories show that the best-response dynamics lead players to converge to the MPE. Note that the intersection of the best-response curves and the equilibrium in the second period do not coincide, because the restriction on the battery size limits players to discharging a total amount in the second period that is below the equilibrium level.

The observable-types setting can be interpreted as a complete-information benchmark for the private-types model: it corresponds to augmenting the private-information game with a public signal that reveals $(\theta,S_t)$ before play in each period. In that case, players' posteriors are degenerate and the period-by-period optimization that underlies Markov perfect equilibrium coincides with the complete-information best-response analysis in Proposition~\ref{prop:market-power}. In the private-types game, players instead choose $b_{it}$ based on their private history and beliefs about $(\theta_{-i},S_{-i,t})$. Nevertheless, the market power effects identified in Proposition~\ref{prop:market-power} translate directly to this case: conditional on any information set at which a player is an interior exporter, her marginal incentive accounts for the effect of $q_{it}$ on the producer price through $\partial_Q p_p$, so exports are shaded relative to the corresponding price-taking benchmark that removes own-price-impact terms from the stationarity condition. Similarly, if the marginal unit stored today is expected to be exported in a future period, the dynamic analogue arises because future sales are valued at (expected) marginal revenue rather than (expected) price. Thus, while Proposition~\ref{prop:market-power} provides sharp state-by-state comparative statics under full observability, in the private-information game the same mechanisms operate up to the approximation error $\varepsilon$.

\section{Computational framework}\label{sec:computation}

\subsection{Multi-agent differentiable market-clearing}\label{subsec:madh}

Each controllable agent~$i$ is represented by a feedforward neural network $f^{\phi^i}$ mapping her observation (battery state, time, demand, PV generation, and a price signal) to a battery action. Training proceeds by unrolling the full $T$-period trajectory and computing policy gradients via backpropagation through the entire computational graph. The method builds on the deep hedging paradigm of \citet{buehler2019deep} and is developed in detail in \cite{EschenbaumGreberSzehr2026}. Because the platform's clearing rule is a known, differentiable mapping from aggregate quantities to prices, the gradient $\nabla_{\phi^i}\bar{C}^i$ can be computed exactly through the chain $b_{it} \mapsto x_{it} \mapsto (M_t, Q_t) \mapsto (p_c, p_p) \mapsto c_{it}$, providing each agent with exact information about her price impact. At each period, agents' actions and market prices are determined jointly through a differentiable tatonnement procedure that iterates price estimates to convergence within the automatic differentiation graph (Appendix~\ref{app:training}).

The method provides a controlled switch between two learning modes. In \emph{strategic} mode, the gradient propagates through the clearing rule, so each agent's update accounts for her effect on prices; per-agent gradients ensure each agent optimises her own cost, not the system cost. In \emph{price-taking} mode, a stop-gradient operator detaches prices from the graph: agents optimise while treating prices as exogenous, corresponding to the price-taking benchmark of Section~\ref{sec:benchmark}. The two modes share identical architectures and hyperparameters; the only difference is whether the gradient passes through the price formation step, providing a controlled ablation for quantifying market power. As a welfare benchmark, we also train a centralised network minimising total grid settlement cost $\sum_t C_t(Y_t)$. Architecture, hyperparameters, and convergence diagnostics are reported in Appendix~\ref{app:computational}.

Convergence to Nash equilibrium is assessed via unilateral deviation regret. For each controllable agent, we freeze all other agents' policies and retrain a single agent.\footnote{Specifically, we freeze other players and retrain agent $i$ for $350$ episodes and then compute $\varepsilon_i := c_i^{\mathrm{eq}} - c_i^{\mathrm{BR}}$.} We train on $1{,}000$ stochastic days and evaluate on $100$ held-out days, averaging over five seeds. We report four metrics: the $\varepsilon$-regret; the withholding ratio $w := (q^{\mathrm{PT}} - q^{\mathrm{S}}) / q^{\mathrm{PT}}$; the average daily grid cost $\bar C := (1/D)\sum_{d} \sum_{t} C_t(Y_t^d)$; and per-agent payoffs (i.e., cost).

\subsection{Data and calibration}\label{subsec:data}

Time is discretized into $T=24$ hourly periods per day ($\Delta t=1\,$h). The baseline community consists of six agent types (Table~\ref{tab:agents}): three controllable types with batteries (prosumer~A, large prosumer~B, pure storage~C) and three passive types (PV generator~D, small PV~E, consumer~F). The battery is lossless ($S_{i,t+1}=S_{it}+b_{it}$) with charge/discharge rate $\bar b_i = \bar S_i / 3$ and initial state $S_{i1}=0$.\footnote{During training, initial battery states are drawn uniformly from $[0, \bar S_i]$ for exploration; all reported results use $S_{i1}=0$.} Although the model sets $V_{i,T+1}\equiv 0$, the computation includes a terminal value $V_{i,T+1} = P^E \cdot S_{i,T+1}$ to prevent degenerate end-of-day discharge; since this is linear in $S_{i,T+1}$, the results of Propositions~\ref{prop:private-existence} and~\ref{prop:market-power} are unaffected.

\begin{table}[h!]
\centering
\caption{Agent types in the baseline platform. }
\label{tab:agents}
\begin{tabular}{@{}llcccc@{}}
\toprule
Type & Role & PV (kWp) & Battery (kWh) & Demand (kWh/d) & Controllable \\
\midrule
A & Prosumer       & 8  & 8  & 13 & Yes \\
B & Large prosumer & 11 & 14 & 15 & Yes \\
C & Pure storage   & 0  & 10 & 0  & Yes \\
D & PV generator   & 8  & 0  & 10 & No  \\
E & Small PV       & 3  & 0  & 12 & No  \\
F & Consumer       & 0  & 0  & 14 & No  \\
\bottomrule
\end{tabular}
\end{table}

Table \ref{tab:agents} shows the parameterization in our baseline platform setting. We distinguish between six types of prosumers for tractability. Three possess controllable assets, in our baseline setting batteries, and three do not. Note that having pure consumers on the platform is important to study the market power incentives, as evening (net) demand for electricity is necessary for platform prices to be able to react to (net) sellers decisions. We fix the charge/discharge rate at $\bar b_i = \bar S_i / 3$.

PV generation follows a sinusoidal clear-sky curve (sunrise 6h, sunset 18h) scaled by agent-specific peak capacity $\bar s_i$ and perturbed by a mean-reverting AR(1) process with random cloud dips. Demand follows a composite time-of-use profile (morning and evening peaks, low midday trough) perturbed by an AR(1) process with random activity pulses and log-normal daily variation ($\approx \pm 20\%$). Full stochastic model parameters are in Appendix~\ref{app:computational}. Figure~\ref{fig:pv-demand-patterns} shows example patterns.

\begin{figure}[h!]
    \centering
    \begin{minipage}{0.48\textwidth}
        \centering
        \includegraphics[width=\linewidth]{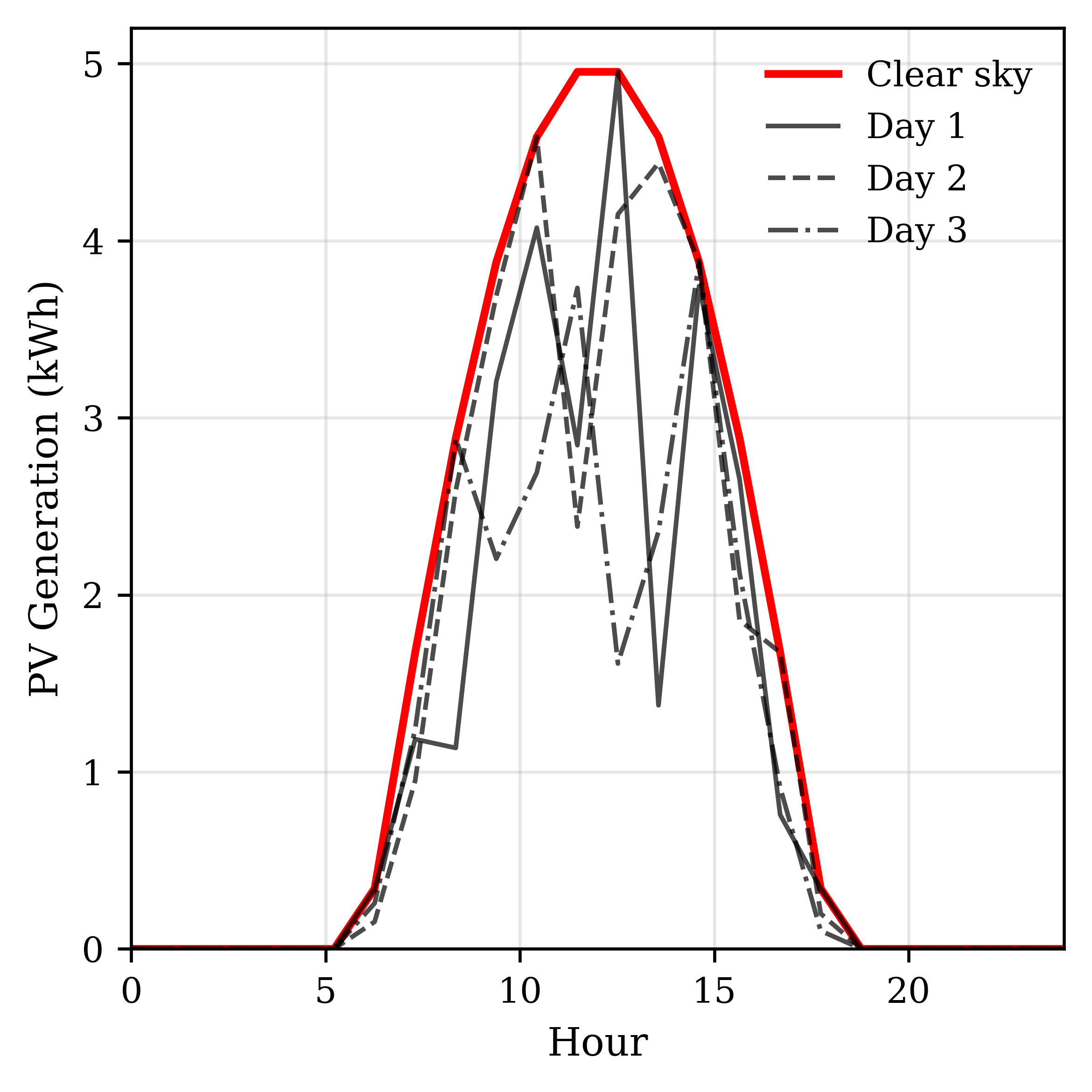}
        \caption*{(a) PV generation $s_i^t$}
    \end{minipage}\hfill
    \begin{minipage}{0.48\textwidth}
        \centering
        \includegraphics[width=\linewidth]{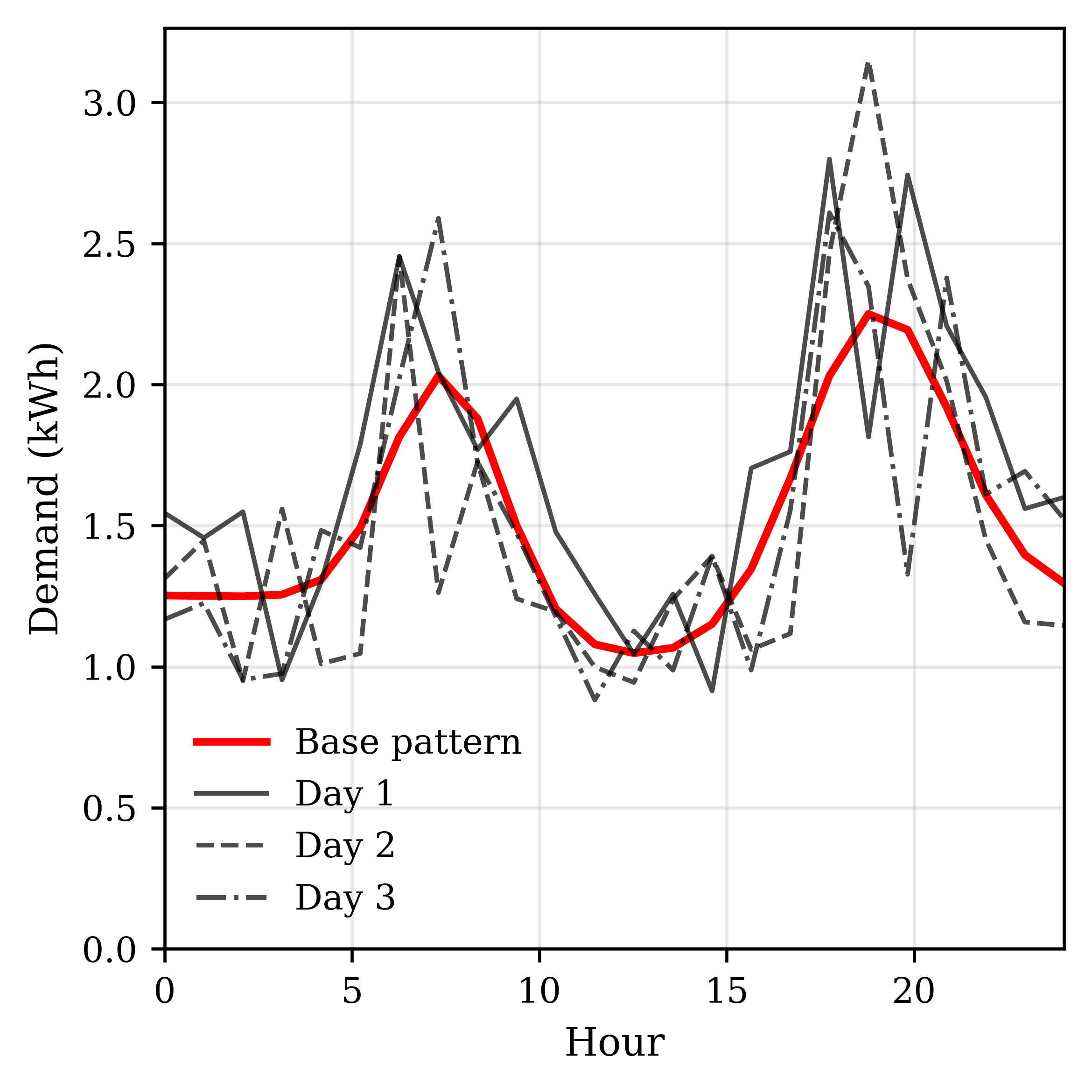}
        \caption*{(b) Demand $d_i^t$}
    \end{minipage}
    \caption{PV generation and demand patterns ($T=24$) with three example stochastic days.}
    \label{fig:pv-demand-patterns}
\end{figure}

Figure \ref{fig:pv-demand-patterns} shows the photovoltaic (PV) and demand underlying patterns and three realized example days for one agent. The actual realizations vary depending on the stochastic perturbations, as discussed above.

\subsection{Pricing mechanisms}\label{subsec:mechanisms}

The platform's pricing rule maps aggregate supply $Q_t$ (total selling power) and demand $M_t$ (total buying power) to a producer price $p_p$ and consumer price $p_c$, both within the tariff corridor $[P^E, P^I]$ where $P^E = 0.05\,$USD/kWh (feed-in tariff) and $P^I = 0.30\,$USD/kWh (retail tariff). We implement three mechanisms that are typical in the literature.

The baseline mechanism is the supply-demand ratio (SDR) applied for example in \citep{liu2017energy}. Formally, it is given by $\rho_t := Q_t / M_t$. The producer (sell) and consumer (buy) prices are, respectively
\begin{alignat*}{2}
p_p^{\mathrm{SDR}} &=
\left\{
\begin{array}{@{}l@{}}
P^E \\[4pt]
\dfrac{P^E \cdot P^I}{(P^I - P^E)\,\rho_t + P^E}
\end{array}
\right.
&\qquad
\begin{array}{@{}l@{}}
\text{if } \rho_t > 1, \\[4pt]
\text{if } \rho_t \le 1,
\end{array}
\\[8pt]
p_c^{\mathrm{SDR}} &=
\left\{
\begin{array}{@{}l@{}}
P^E \\[4pt]
p_p^{\mathrm{SDR}} \cdot \rho_t + P^I \cdot (1 - \rho_t)
\end{array}
\right.
&\qquad
\begin{array}{@{}l@{}}
\text{if } \rho_t > 1, \\[4pt]
\text{if } \rho_t \le 1.
\end{array}
\end{alignat*}
When supply exceeds demand ($\rho_t > 1$), all volume clears at the feed-in tariff. When demand exceeds supply, the producer price rises hyperbolically toward $P^I$ and the consumer price is a weighted average of the internal and external prices.

An alternative mechanism we consider is the mid-market rate (MMR), used in \citep{long2017peer}. $p_{\mathrm{mid}} := (P^E + P^I)/2 = 0.175\,$USD/kWh. The short side of the market faces $p_{\mathrm{mid}}$; the long side is adjusted toward the external tariff via pro-rata weighting:
\begin{align*}
p_p^{\mathrm{MMR}} &=
\begin{cases}
p_{\mathrm{mid}} & \text{if } M_t \ge Q_t, \\
(M_t \cdot p_{\mathrm{mid}} + (Q_t - M_t)\cdot P^E) / Q_t & \text{if } M_t < Q_t,
\end{cases}
\\[6pt]
p_c^{\mathrm{MMR}} &=
\begin{cases}
p_{\mathrm{mid}} & \text{if } M_t \le Q_t, \\
(Q_t \cdot p_{\mathrm{mid}} + (M_t - Q_t)\cdot P^I) / M_t & \text{if } M_t > Q_t.
\end{cases}
\end{align*}
The MMR mechanism produces less price variation than SDR, as the clearing price depends only on which side is short, not on the magnitude of the imbalance.

Finally, we also consider a simple linear pricing mechanism. This mechanism implements an inverse-demand clearing price
\[
p^{\mathrm{LIN}} = p_{\mathrm{mid}} - \kappa\,(Q_t - M_t), \qquad \text{clamped to } [P^E,\, P^I],
\]
where $\kappa = (P^I - P^E)/10 = 0.025\,$USD/kWh$^2$ spans the tariff corridor over a $\pm 5\,$kWh imbalance. The long side clears pro-rata at the internal price, with the remainder settled at the external tariff:
\begin{align*}
p_p^{\mathrm{LIN}} &= \min(M_t/Q_t,\,1)\cdot p^{\mathrm{LIN}} + (1 - \min(M_t/Q_t,\,1))\cdot P^E, \\
p_c^{\mathrm{LIN}} &= \min(Q_t/M_t,\,1)\cdot p^{\mathrm{LIN}} + (1 - \min(Q_t/M_t,\,1))\cdot P^I.
\end{align*}
The linear mechanism has a particularly simple structure. The clearing price is a smooth, monotone function of net supply, making it the most transparent for identifying market power effects.

Note that all three mechanisms produce prices in $[P^E, P^I]$ and satisfy Assumption~\ref{ass:prices1}. All three mechanisms satisfy the monotonicity conditions of Assumption~\ref{ass:prices2}: an increase in aggregate supply weakly reduces the producer price, and an increase in aggregate demand weakly raises the consumer price, with strict monotonicity whenever the price is in the interior of the tariff corridor $[P^E, P^I]$. The convexity and concavity conditions of Assumption~\ref{ass:prices2} hold globally for the MMR mechanism. For the SDR and linear mechanisms, the curvature conditions hold within each smooth pricing regime, but can fail at regime boundaries: the SDR mechanism has a convex kink in the export-revenue function at $Q=M$ and loses import-cost convexity at extreme supply-demand ratios, while the linear mechanism's price clamping at $P^E$ and $P^I$ creates kinks that can break both global concavity and convexity. For both mechanisms, we verify in Appendix~\ref{app:verification} that the total objective (stage payoff plus continuation value) remains concave in $b_{it}$ on the attainable set. The concavity of the equilibrium continuation value function required for pure-strategy MPE existence is also verified in that appendix.

\section{Computational Results}\label{sec:results}

We present computational results from three treatments designed to test the theoretical predictions of Section~\ref{sec:benchmark}: a baseline ordering that confirms the cost ranking of Proposition~\ref{prop:market-power}, a competition treatment that varies the number of storage operators, and a heterogeneity treatment that compares pricing mechanisms across the six agent types. All results are averaged over five random seeds; error bars or bands show $\pm 1$ standard deviation across seeds. The baseline community is described in Table~\ref{tab:agents}. For each treatment, we train price-taking and strategic benchmarks---and, where noted, a social planner---using the computational framework of Section~\ref{sec:computation}. Full experimental configurations and hyperparameters are reported in Appendix~\ref{app:computational}.


\subsection{Baseline ordering}\label{subsec:baseline}

We fix the SDR mechanism with private information and train all three regimes on the baseline community (Table~\ref{tab:agents}). Table~\ref{tab:baseline-results} reports grid cost, community cost, and the withholding ratio.

\begin{table}[h!]
\centering
\caption{Baseline cost ordering (SDR mechanism, five seeds). Standard deviations in parentheses. Grid cost is the community's total settlement with the external grid; since internal transfers cancel, grid cost equals the welfare-relevant total surplus (up to sign).}
\label{tab:baseline-results}
\begin{tabular}{@{}lc@{}}
\toprule
Regime & Grid cost (USD/day) \\
\midrule
Social planner    & 5.67\;(0.67) \\
Price-taking      & 5.71\;(0.58) \\
Strategic         & 6.07\;(0.45) \\
\midrule
Strategic $-$ Price-taking & $+$0.36\;[$+$6.3\%] \\
\addlinespace
Withholding ratio $w$ & 0.029\;(0.039) \\
\bottomrule
\end{tabular}
\end{table}

Grid cost follows the ordering predicted by Proposition~\ref{prop:market-power}: Planner (USD~5.67/day) $\leq$ Price-taking (5.71) $<$ Strategic (6.07). The strategic premium of USD~0.36/day ($6.3\%$) is the aggregate welfare loss from export withholding ($w = 0.029$). The planner--price-taking gap is small (USD~0.04/day), indicating that the primary source of inefficiency is strategic behaviour. Since internal transfers between players and the platform cancel, grid cost is the welfare-relevant measure of total surplus; per-agent costs, reported in subsequent tables, capture only the distribution of surplus across participants.

\begin{figure}[h!]
\centering
\includegraphics[width=0.98\textwidth]{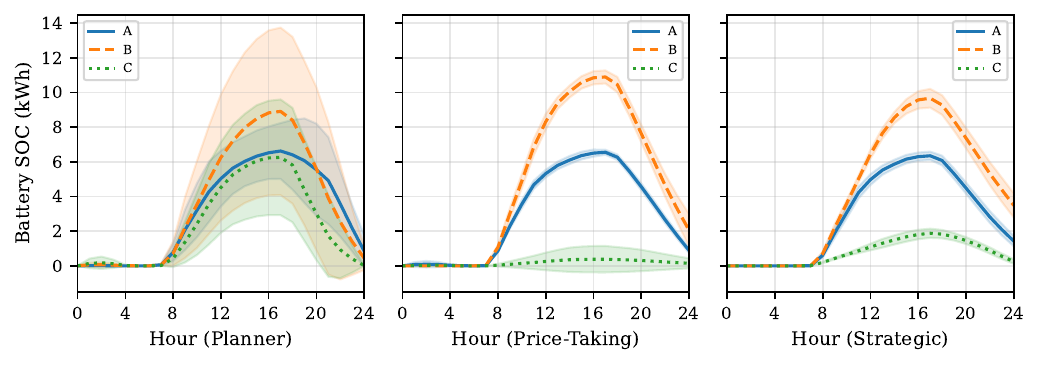}
\caption{Average battery SOC over 24 hours by regime (SDR, five seeds, $\pm 1$ s.d.). The planner uses all three batteries; storage agent~C is idle under price-taking. Strategic agents retain higher terminal SOC.}
\label{fig:baseline-battery}
\end{figure}

Figure~\ref{fig:baseline-battery} shows battery SOC over 24 hours. All regimes follow the expected diurnal cycle: charge during midday PV surplus, discharge during the evening peak. The planner uses all three batteries aggressively, including agent~C (peak SOC $6.3\,$kWh of $10\,$kWh). Under price-taking, C is nearly idle (peak $0.4\,$kWh): without own generation or demand, the bid--ask spread does not compensate enough round-trip trading. 

\begin{figure}[h!]
\centering
\includegraphics[width=0.98\textwidth]{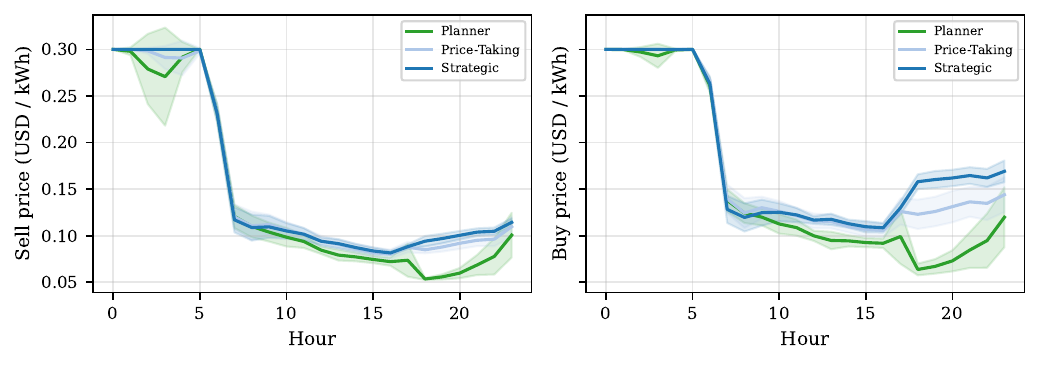}
\caption{Average sell and buy prices over 24 hours (SDR, five seeds, $\pm 1$ s.d.). Strategic play raises evening buy prices by $\approx 0.03\,$USD/kWh relative to price-taking.}
\label{fig:baseline-prices}
\end{figure}

Figure~\ref{fig:baseline-prices} shows market prices. Nighttime prices equal the retail tariff ($P^I = 0.30$); midday PV surplus pushes prices to USD~$0.08$--$0.13$/kWh across regimes with minimal inter-regime differences. The key distinction appears in the evening (hours $18$--$23$): strategic buy prices average USD~$0.16$/kWh versus USD~$0.13$/kWh under price-taking---a premium of USD~$0.03$/kWh ($+23\%$). This is the price-side reflection of quantity withholding: restricted evening exports shift the supply--demand ratio, raising the clearing price at the expense of passive consumers. The planner pushes evening prices down to USD~$0.07$--$0.12$/kWh through aggressive discharge.

Table~\ref{tab:regret} reports the unilateral-deviation
$\varepsilon$-regret for each controllable agent. The regret is computed by freezing all
opponents' policies at their strategic equilibrium values and retraining the focal agent
as a best-responder for 350 episodes (see Section~\ref{subsec:madh}). All regrets are
below USD~$0.01$/day in absolute value, confirming that the strategic outcome is an
approximate Nash equilibrium: no agent can reduce her daily cost by more than one cent
through unilateral deviation.

\begin{table}[h!]
\centering
\caption{Unilateral-deviation $\varepsilon$-regret (USD/day, SDR, five seeds).
$\varepsilon_i := C_i^{\mathrm{eq}} - C_i^{\mathrm{BR}}$; positive values mean the
equilibrium cost exceeds the best response (agent could improve), negative values mean
the best-responder fared worse. Standard deviations in parentheses.}
\label{tab:regret}
\begin{tabular}{@{}llcc@{}}
\toprule
Agent & Type & $\varepsilon_i$ (USD/day) & $|\varepsilon_i|/|C_i^{\mathrm{eq}}|$ \\
\midrule
A & Prosumer (8\,kWh)        & $+$0.005\;(0.004)  & 0.7\% \\
B & Large prosumer (14\,kWh) & $-$0.002\;(0.005)  & 0.5\% \\
C & Pure storage (10\,kWh)   & $+$0.006\;(0.008)  & --- \\
\bottomrule
\end{tabular}
\end{table}


\subsection{Competition and market power}\label{subsec:competition}

We vary the number of symmetric storage agents $N \in \{1,2,3\}$, holding total capacity fixed at $10\,$kWh ($\bar S = 10/N$ per agent, $\bar b = \bar S/3$). Two passive agents (D, E) provide supply and demand. The mechanism is SDR; all results average five seeds.

\begin{table}[h!]
\centering
\caption{Grid cost, withholding, and price impact by number of storage agents $N$ (SDR, five seeds). Standard deviations in parentheses. Price impact computed via automatic differentiation through the clearing rule.}
\label{tab:competition-results}
\begin{tabular}{@{}lccccc@{}}
\toprule
 & \multicolumn{2}{c}{Grid cost (USD/day)} & Gap & Withholding & Price impact \\
$N$ & Strategic & Price-taking & (\%) & $w$ & $|\partial p / \partial b_i|$ \\
\midrule
1 & 2.77\;(0.18) & 2.02\;(0.23) & $+$37.1 & 0.707\;(0.036) & 0.198\;(0.002) \\
2 & 2.45\;(0.23) & 2.10\;(0.21) & $+$16.7 & 0.342\;(0.068) & 0.162\;(0.002) \\
3 & 2.30\;(0.23) & 2.28\;(0.24) & $+$0.9  & $-$0.005\;(0.214) & 0.140\;(0.002) \\
\bottomrule
\end{tabular}
\end{table}

The results show the expected patterns. With a single storage agent, the strategic grid cost exceeds the price-taking benchmark by $37\%$ and the agent withholds $71\%$ of competitive exports. A duopoly reduces the gap to $17\%$ (withholding $34\%$). With three storage agents, strategic and price-taking outcomes become statistically indistinguishable. The largest welfare gain comes from the first entrant, consistent with classic oligopoly theory. Figure~\ref{fig:competition-battery} illustrates the mechanism directly. Total battery utilisation under strategic play rises monotonically in $N$. Whereas a single storage agent charges to a peak SOC of only $1.9\,$kWh (out of a total of $10\,$kWh), two competing agents reach $3.1\,$kWh and three competing agents reach $4.1\,$kWh. Competition therefore erodes each agent's incentive to withhold, pushing battery behaviour toward the competitive outcome.

\begin{figure}[h!]
\centering
\includegraphics[width=0.98\textwidth]{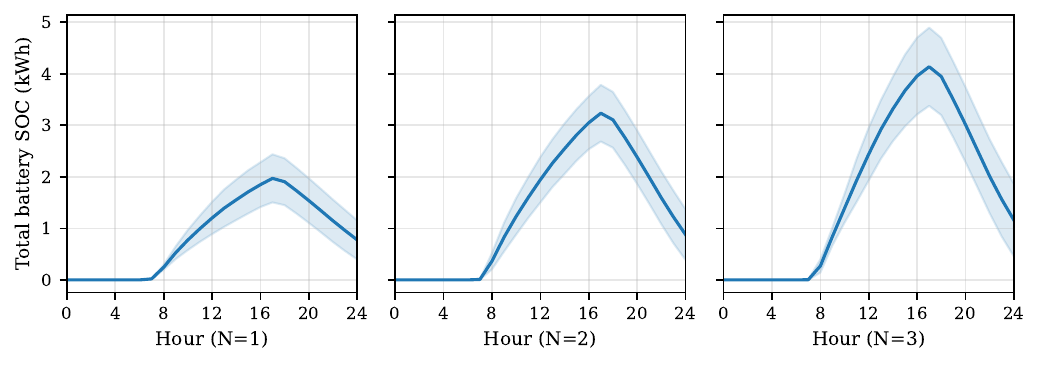}
\caption{Total battery SOC over 24 hours under strategic play, by number of competing storage agents $N$ (SDR, five seeds, $\pm 1$ s.d.). Battery utilisation increases monotonically with competition.}
\label{fig:competition-battery}
\end{figure}




Table~\ref{tab:competition-distributional} shows the resulting payoff (or cost) by agent type. With a single storage agent (a monopolist), this  agent earns USD~$0.12$/day in the approximate equilibrium compared to USD~$0.04$/day when ignoring their own price impact. This additional rents is redistributed from inflexible agents whose costs rise by USD~$0.36$ (D) and USD~$0.46$ (E). The monopolist captures only USD~$0.08$/day---roughly $10\%$ of passive agents' losses---with the remaining $90\%$ dissipated as deadweight through increased grid reliance. Under a duopoly, this rent reduces by half. With three competing storage agents, the distributional effects vanish.

\begin{table}[h!]
\centering
\caption{Per-agent costs by market structure (USD/day, five seeds). Battery cost is the per-agent average across $N$ symmetric agents. $\Delta = \text{Strategic} - \text{Price-taking}$; negative values indicate rent extraction. Standard deviations in parentheses.}
\label{tab:competition-distributional}
\begin{tabular}{@{}clccc@{}}
\toprule
$N$ & Agent & Price-taking & Strategic & $\Delta$ \\
\midrule
1 & Battery ($\times 1$)     & $-$0.04\;(0.02) & $-$0.12\;(0.01) & $-$0.08 \\
  & PV generator (D)         &    0.42\;(0.13)  &    0.77\;(0.09)  & $+$0.36 \\
  & Small PV (E)             &    1.62\;(0.12)  &    2.07\;(0.08)  & $+$0.46 \\
\addlinespace
2 & Battery ($\times 2$, avg) & $-$0.03\;(0.03) & $-$0.05\;(0.01) & $-$0.02 \\
  & PV generator (D)          &    0.47\;(0.11)  &    0.63\;(0.12)  & $+$0.16 \\
  & Small PV (E)              &    1.68\;(0.11)  &    1.88\;(0.12)  & $+$0.20 \\
\addlinespace
3 & Battery ($\times 3$, avg) & $-$0.02\;(0.04) & $-$0.03\;(0.01) & $\approx 0$ \\
  & PV generator (D)          &    0.55\;(0.12)  &    0.55\;(0.12)  & $\approx 0$ \\
  & Small PV (E)              &    1.79\;(0.13)  &    1.78\;(0.12)  & $\approx 0$ \\
\bottomrule
\end{tabular}
\end{table}

\subsection{Heterogeneity and mechanism design}\label{subsec:heterogeneity}
 
We now exploit the full heterogeneity of the baseline community to compare the distributional consequences of strategic play across pricing mechanisms. Table~\ref{tab:heterogeneity-results} reports the per-agent cost difference, where  we define $\Delta = \text{Strategic} - \text{Price-taking}$ in all three mechanisms. As before, we provide the results across five seeds.
 
\begin{table}[h!]
\centering
\caption{Per-agent strategic premium $\Delta_i = C^S_i - C^{PT}_i$ (USD/day, five seeds). Negative values indicate that the agent benefits from strategic play. The grid cost row reports the difference in total grid settlement cost. Standard deviations in parentheses.}

\label{tab:heterogeneity-results}
\begin{tabular}{@{}clccc@{}}
\toprule
 & & \multicolumn{3}{c}{$\Delta$: Strategic $-$ Price-taking (USD/day)} \\
\cmidrule(l){3-5}
Agent & Type & SDR & MMR & LINEAR \\
\midrule
A & Prosumer (8\,kWh)        & $-$0.018\;(0.009) & $-$0.002\;(0.023) & $-$0.054\;(0.028) \\
B & Large prosumer (14\,kWh) & $-$0.005\;(0.020) & $-$0.007\;(0.035) & $-$0.028\;(0.033) \\
C & Pure storage (10\,kWh)   & $-$0.032\;(0.008) &    0.000\;(0.000) & $-$0.001\;(0.002) \\
D & PV generator             & $+$0.120\;(0.053) & $+$0.007\;(0.009) & $+$0.042\;(0.015) \\
E & Small PV                 & $+$0.134\;(0.063) & $+$0.001\;(0.006) & $+$0.025\;(0.010) \\
F & Consumer                 & $+$0.139\;(0.070) & $-$0.001\;(0.007) & $+$0.017\;(0.009) \\
\midrule
  & Grid cost $\Delta$       & $+$0.491\;(0.182) & $-$0.011\;(0.023) & $+$0.042\;(0.038) \\
  & Withholding $w$          &    0.040\;(0.029)  & $-$0.003\;(0.008) &    0.033\;(0.008) \\
\bottomrule
\end{tabular}
\end{table}
 
The distributional pattern is consistent across mechanisms---battery owners gain, passive agents lose---but the magnitudes differ by an order of magnitude. Under SDR, the pattern follows directly from the baseline withholding mechanism: battery-equipped agents lower their costs by USD~$0.005$--$0.032$/day, while passive agents face increases of USD~$0.12$--$0.14$/day. The total grid cost increase (USD~$0.49$/day) exceeds battery owners' combined savings (USD~$0.06$/day) by a factor of six, so the welfare loss is predominantly deadweight.
 
The choice of pricing mechanism has a substantial effect on these findings. Under MMR, all per-agent gains from equilibrium play are indistinguishable from zero: the fixed mid-market rate eliminates the price sensitivity that strategic agents exploit ($w \approx 0$). With a linear price mechanism, withholding is comparable to SDR ($w = 0.033$) and battery owners benefit from strategic play, but the grid cost increase is negligible ($\Delta \approx 0$). The linear mechanism's smooth price response allows individual agents to exercise market power, but the surplus they extract comes almost entirely from redistribution rather than deadweight loss---grid settlement cost is roughly identical even as the division of surplus on the platform shifts toward battery owners.

\subsection{Discussion: magnitudes, information, and platform value}\label{subsec:platform-value}
 
We conclude the results section by assessing the value of the platform for participating players and examining how information design affects our results.
We assume that without the platform, each agent settles exclusively with the external grid at the retail tariff $P^I = 0.30\,$USD/kWh (imports) and feed-in tariff $P^E = 0.05\,$USD/kWh (exports). Table~\ref{tab:platform-value} projects annual costs with and without the platform (we assume the platform operates under the SDR mechanism).
 
\begin{table}[h!]
\centering
\caption{Annual cost projection (USD/year, SDR, five seeds). ``No platform'' assumes exclusive grid settlement at regulated tariffs with self-consumption of own PV. Platform saving $=$ No platform $-$
Price-taking. Market power cost $=$ Strategic $-$ Price-taking.}
\label{tab:platform-value}
\begin{tabular}{@{}lccccc@{}}
\toprule
& No platform & \multicolumn{2}{c}{Platform (annual)} & Platform & Market power \\
Agent & (grid only) & Price-taking & Strategic & saving & cost \\
\midrule
F (consumer, 14\,kWh/d)  & 1{,}585 & 921 & 971 & 664 & 51 \\
E (small PV, 12\,kWh/d)  & 884 & 588 & 637 & 295 & 49 \\
D (PV gen, 10\,kWh/d)    & 375 & 140 & 184 & 235 & 44 \\
Community (6 agents)      & 4{,}102 & 2{,}071 & 2{,}250 & 2{,}080 & 179 \\
\bottomrule
\end{tabular}
\end{table}

The pure consumer~F provides the cleanest comparison: without a platform, F settles exclusively with the grid at the retail tariff, paying USD~$1{,}585$/year. Platform trade reduces this to USD~$921$/year ($-42\%$), and strategic behaviour adds USD~$51$/year back---roughly $8\%$ of the platform's saving. The PV generator~D benefits most from the platform in relative terms ($-63\%$), because it can sell surplus generation internally at prices well above the feed-in tariff. The community-wide cost of market power is USD~$179$/year, of which battery owners capture roughly USD~$21$ in rents; the remaining $89\%$ is dissipated as deadweight loss through increased grid reliance. The platform is thus overwhelmingly valuable despite strategic distortions, but the welfare loss falls disproportionately on passive agents, reinforcing the case for careful mechanism design.
 
We further compare two information structures under the SDR mechanism (as before with five seeds): \emph{private information}, where each agent observes only her own type and battery state, and \emph{price observation}, where agents additionally observe past clearing prices. Table~\ref{tab:info-results} shows community cost and withholding under each structure.\footnote{The private-information treatment in Table~\ref{tab:info-results} uses an independent training run; the small difference from Table~\ref{tab:baseline-results} reflects cross-run variation.}
 
\begin{table}[h!]
\centering
\caption{Information design treatment (SDR, five seeds). Standard deviations in parentheses. Strategic premium is the percentage increase in grid cost: $(\bar C^S - \bar C^{PT})/\bar C^{PT}$.}
\label{tab:info-results}
\begin{tabular}{@{}lcccc@{}}
\toprule
Information & \multicolumn{2}{c}{Grid cost (USD/day)} & Strategic & Withholding \\
structure & Strategic & Price-taking & premium & $w$ \\
\midrule
Private           & 6.22\;(0.55) & 5.69\;(0.66) & 9.3\% & 0.051\;(0.025) \\
Price observation  & 6.75\;(0.94) & 5.64\;(0.57) & 19.7\% & 0.038\;(0.070) \\
\bottomrule
\end{tabular}
\end{table}
 
Price observation lowers price-taking grid cost by $0.9\%$ (USD~$5.69 \to 5.64$/day), showing the efficiency value of information disclosure, but it raises the strategic grid cost by $8.5\%$ (USD~$6.22 \to 6.75$/day), more than doubling the strategic premium from $9.3\%$ to $19.7\%$. The welfare effect of price transparency thus depends on whether agents exercise market power: richer signals help competitive agents coordinate with market conditions but also enable strategic agents to exploit their price impact more effectively. This echoes the classical IO debate on whether information facilitates or disciplines strategic behaviour \citep[cf.][]{stigler1964theory,calvano2020artificial}. We note that the cross-seed variance under strategic price observation is high (s.d.~$= 0.84$), so these magnitudes should currently be interpreted cautiously.

Experiments replicating the six-agent community $k$ times ($k \in \{1,2,3\}$, single seed) suggest that market power diminishes rapidly with community size. The strategic premium falls from $9.3\%$ at $N=6$ to near zero at $N=12$--$18$, consistent with the $1/N$ Cournot convergence in the competition treatment. Internal matching rates are roughly stable under replication of identical types ($\approx 12\%$ of trade) implying that platforms with genuinely diverse consumption and generation profiles---where stochastic variation across agents creates natural surpluses and deficits that partially substitute for battery-mediated arbitrage---would likely yield higher matching rates and are a natural direction for future work.


\section{Conclusion}\label{sec:conclusion}

This paper studies market power and platform design in decentralized electricity trading. We develop a dynamic game in which photovoltaic (PV) and battery-equipped prosumers trade on a platform that precommits to a pricing rule, and we characterize how the interaction between mechanism design, information design, and storage ownership shapes equilibrium outcomes.
 
Three findings stand out. First, strategic behaviour by battery owners generates a welfare loss that is predominantly deadweight: under the supply-demand ratio mechanism, strategic play raises grid cost by roughly 6\%, but battery owners capture only a small fraction of passive agents' losses. The distortion falls most heavily on passive agents who lack storage. Second, competition among storage operators is the most effective discipline on market power. Moving from monopoly to triopoly eliminates nearly all strategic distortion, consistent with standard Cournot convergence. This suggests that policies encouraging fragmented storage ownership, such as subsidies for residential batteries rather than community-scale installations controlled by a single operator, can have beneficial effects beyond the direct value of storage capacity. Third, the choice of pricing rule has first-order effects on the magnitude of distortions. A mid-market rate mechanism virtually eliminates strategic incentives by removing the price sensitivity that agents exploit, while a linear mechanism permits individual withholding but confines the resulting transfers largely to redistribution rather than deadweight loss. Platform operators thus have a concrete design lever to mitigate market power.
 
Our results also carry a more reassuring message. The platform is overwhelmingly valuable despite these distortions: a passive consumer saves roughly 40\% on annual electricity costs relative to exclusive grid settlement, and market power claws back only about 8\% of that saving. The policy question is therefore not whether to allow decentralized trading, but how to design the institutions that govern it.
 
Several limitations of the current analysis point to directions for future work. First, our community comprises six agents with stylized demand and generation profiles. Scaling the analysis to larger and more heterogeneous communities, where stochastic variation across agents creates natural surpluses and deficits that partially substitute for battery-mediated arbitrage, would clarify how quickly market power dissipates with community size. Preliminary results suggest rapid convergence, but a systematic treatment with diverse agent compositions remains open. Second, we consider three pricing mechanisms that are typical in the literature, but the design space is much richer. Exploring mechanisms that condition on richer signals, that incorporate dynamic elements such as intertemporal price linkages, or that allow the platform to optimise its rule endogenously would connect this work more directly to the mechanism design literature. Third, our model is a finite-horizon game with a fixed set of players. Extending the framework to an infinite-horizon setting with discounting would allow the analysis of long-run dynamics, including the possibility that patient agents sustain collusive outcomes through repeated interaction. Whether the Cournot-style withholding we document can be amplified by tacit coordination, and whether platform design can forestall such outcomes, are important open questions. Finally, our computational framework treats the pricing rule as given and compares mechanisms ex post. An ambitious extension would be to let the platform learn its own mechanism, along the lines of recent work on differentiable mechanism design, subject to regulatory constraints such as the tariff corridor and budget balance.
 
More broadly, our analysis illustrates that the industrial economics of electricity markets extends naturally to the emerging setting of prosumer trading platforms. The same forces that shape strategic behaviour in wholesale markets, namely market concentration, information asymmetries, and the design of price formation rules, operate at the prosumer scale. As decentralized generation and storage continue to expand, understanding these forces will be essential for designing platforms that deliver on the promise of local energy trading.

\newpage
\bibliographystyle{aer}
\bibliography{references}

@article{myerson2020epsilon,
author = {Roger B. Myerson AND Philip J. Reny},
title = {Perfect Conditional {$\varepsilon$}-Equilibria of Multi-Stage Games with Infinite Sets of Signals and Actions},
journal = {Econometrica},
volume = {88},
number = {2},
pages = {495-531},
doi = {https://doi.org/10.3982/ECTA13426},
url = {https://onlinelibrary.wiley.com/doi/abs/10.3982/ECTA13426},
eprint = {https://onlinelibrary.wiley.com/doi/pdf/10.3982/ECTA13426},
year = {2020}
}

@article{he2021energypawn,
  author  = {He, Li and Liu, Yuanzhi and Zhang, Jie},
  title   = {Peer-to-peer energy sharing with battery storage: Energy pawn in the smart grid},
  journal = {Applied Energy},
  volume  = {297},
  pages   = {117129},
  year    = {2021},
  doi     = {10.1016/j.apenergy.2021.117129}
}

@article{hoseinpour2024costrecovery,
  author  = {Hoseinpour, Milad and Haghifam, Mahmoud-Reza},
  title   = {Prosumers' cost recovery in peer-to-peer electricity markets},
  journal = {Electric Power Systems Research},
  volume  = {226},
  pages   = {109934},
  year    = {2024},
  doi     = {10.1016/j.epsr.2023.109934}
}

@article{williams2022storage,
  author  = {Williams, Olayinka and Green, Richard},
  title   = {Electricity storage and market power},
  journal = {Energy Policy},
  volume  = {164},
  pages   = {112872},
  year    = {2022},
  doi     = {10.1016/j.enpol.2022.112872}
}

@article{baake2023local,
  author  = {Baake, Pio and Schwenen, Sebastian and von Hirschhausen, Christian},
  title   = {Local Energy Markets},
  journal = {Journal of Industrial Economics},
  volume  = {71},
  number  = {3},
  pages   = {855--882},
  year    = {2023},
  doi     = {10.1111/joie.12338}
}

@article{kastius2022dynamic,
  author  = {Kastius, Alexander and Schlosser, Rainer},
  title   = {Dynamic pricing under competition using reinforcement learning},
  journal = {Journal of Revenue and Pricing Management},
  volume  = {21},
  number  = {1},
  pages   = {50--63},
  year    = {2022},
  doi     = {10.1057/s41272-021-00285-3}
}

@article{EschenbaumGreberSzehr2026,
  title={Differentiable Market Clearing for Multi-Agent Learning in Peer-to-Peer Electricity Trading},
  author={Eschenbaum, Nicolas and Greber, Nicolas and Szehr, Oleg},
  journal={arXiv},
  year={2026}
}

@article{calvano2020artificial,
  title={Artificial intelligence, algorithmic pricing, and collusion},
  author={Calvano, Emilio and Calzolari, Giacomo and Denicolo, Vincenzo and Pastorello, Sergio},
  journal={American Economic Review},
  volume={110},
  number={10},
  pages={3267--3297},
  year={2020},
  publisher={American Economic Association 2014 Broadway, Suite 305, Nashville, TN 37203}
}

@article{liu2017energy,
  title={Energy-sharing model with price-based demand response for microgrids of peer-to-peer prosumers},
  author={Liu, Nian and Yu, Xinghuo and Wang, Cheng and Li, Chaojie and Ma, Li and Lei, Jinyong},
  journal={IEEE Transactions on Power Systems},
  volume={32},
  number={5},
  pages={3569--3583},
  year={2017},
  publisher={IEEE}
}

@article{tushar2021p2p_review,
  title={Peer-to-peer trading in electricity networks: An overview},
  author={Tushar, Wayes and Saha, Tapan Kumar and Yuen, Chau and Smith, David and Poor, H Vincent},
  journal={IEEE transactions on smart grid},
  volume={11},
  number={4},
  pages={3185--3200},
  year={2020},
  publisher={IEEE}
}

@article{le2020peer,
  title={Peer-to-peer electricity market analysis: From variational to generalized Nash equilibrium},
  author={Le Cadre, H{\'e}l{\`e}ne and Jacquot, Paulin and Wan, Cheng and Alasseur, Cl{\'e}mence},
  journal={European Journal of Operational Research},
  volume={282},
  number={2},
  pages={753--771},
  year={2020},
  publisher={Elsevier}
}

@article{andrescerezo2023storing,
  author    = {Andr{\'e}s-Cerezo, David and Fabra, Natalia},
  title     = {Storing Power: Market Structure Matters},
  journal   = {The RAND Journal of Economics},
  year      = {2023},
  volume    = {54},
  number    = {1},
  pages     = {3--53},
  doi       = {10.1111/1756-2171.12429},
}

@article{schill2011strategic,
  author    = {Schill, Wolf-Peter and Kemfert, Claudia},
  title     = {Modeling Strategic Electricity Storage: The Case of Pumped Hydro Storage in {Germany}},
  journal   = {The Energy Journal},
  year      = {2011},
  volume    = {32},
  number    = {3},
  pages     = {59--87},
  doi       = {10.5547/ISSN0195-6574-EJ-Vol32-No3-3},
}

@article{fabra2023market,
  author    = {Fabra, Natalia and Imelda},
  title     = {Market Power and Price Exposure: Learning from Changes in Renewable Energy Regulation},
  journal   = {American Economic Journal: Economic Policy},
  year      = {2023},
  volume    = {15},
  number    = {4},
  pages     = {323--358},
  doi       = {10.1257/pol.20210221},
}

@article{fabra2021energy,
  author    = {Fabra, Natalia},
  title     = {The Energy Transition: An Industrial Economics Perspective},
  journal   = {International Journal of Industrial Organization},
  year      = {2021},
  volume    = {79},
  pages     = {102734},
  doi       = {10.1016/j.ijindorg.2021.102734},
}

@article{sioshansi2014storage,
  author    = {Sioshansi, Ramteen},
  title     = {When Energy Storage Reduces Social Welfare},
  journal   = {Energy Economics},
  year      = {2014},
  volume    = {41},
  pages     = {106--116},
  doi       = {10.1016/j.eneco.2013.09.027},
}

@article{garcia2001strategic,
  author    = {Garcia, Alfredo and Reitzes, James D. and Stacchetti, Ennio},
  title     = {Strategic Pricing When Electricity Is Storable},
  journal   = {Journal of Regulatory Economics},
  year      = {2001},
  volume    = {20},
  number    = {3},
  pages     = {223--247},
  doi       = {10.1023/A:1012281721110},
}

@article{bushnell2003looking,
  author    = {Bushnell, James},
  title     = {A Mixed Complementarity Model of Hydrothermal Electricity Competition in the {Western United States}},
  journal   = {Operations Research},
  year      = {2003},
  volume    = {51},
  number    = {1},
  pages     = {80--93},
  doi       = {10.1287/opre.51.1.80.12800},
}

@article{parag2016prosumer,
  author    = {Parag, Yael and Sovacool, Benjamin K.},
  title     = {Electricity Market Design for the Prosumer Era},
  journal   = {Nature Energy},
  year      = {2016},
  volume    = {1},
  number    = {4},
  pages     = {16032},
  doi       = {10.1038/nenergy.2016.32},
}

@article{mengelkamp2018brooklyn,
  author    = {Mengelkamp, Esther and G{\"a}rttner, Johannes and Rock, Kerstin and Kessler, Scott and Orsini, Lawrence and Weinhardt, Christof},
  title     = {Designing Microgrid Energy Markets: A Case Study: The {Brooklyn Microgrid}},
  journal   = {Applied Energy},
  year      = {2018},
  volume    = {210},
  pages     = {870--880},
  doi       = {10.1016/j.apenergy.2017.06.054},
}

@article{morstyn2018federated,
  author    = {Morstyn, Thomas and Farrell, Niall and Darby, Sarah J. and McCulloch, Malcolm D.},
  title     = {Using Peer-to-Peer Energy-Trading Platforms to Incentivize Prosumers to Form Federated Power Plants},
  journal   = {Nature Energy},
  year      = {2018},
  volume    = {3},
  number    = {2},
  pages     = {94--101},
  doi       = {10.1038/s41560-017-0075-y},
}

@techreport{irena2020p2p,
  author    = {{IRENA}},
  title     = {Innovation Landscape Brief: Peer-to-Peer Electricity Trading},
  institution = {International Renewable Energy Agency},
  year      = {2020},
  address   = {Abu Dhabi},
}

@article{zhang2020cournot,
  author    = {Zhang, Yizhou and Gu, Chenghong and Yan, Xiangning and Li, Furong},
  title     = {Cournot Oligopoly Game-Based Local Energy Trading Considering Renewable Energy Uncertainty Costs},
  journal   = {Renewable Energy},
  year      = {2020},
  volume    = {159},
  pages     = {1117--1127},
  doi       = {10.1016/j.renene.2020.06.062},
}

@article{buehler2019deep,
  author    = {Buehler, Hans and Gonon, Lukas and Teichmann, Josef and Wood, Ben},
  title     = {Deep Hedging},
  journal   = {Quantitative Finance},
  year      = {2019},
  volume    = {19},
  number    = {8},
  pages     = {1271--1291},
  doi       = {10.1080/14697688.2019.1571683},
}

@article{wolfram1999measuring,
  author    = {Wolfram, Catherine D.},
  title     = {Measuring Duopoly Power in the {British} Electricity Spot Market},
  journal   = {American Economic Review},
  year      = {1999},
  volume    = {89},
  number    = {4},
  pages     = {805--826},
  doi       = {10.1257/aer.89.4.805},
}

@article{hortacsu2008understanding,
  author    = {Horta\c{c}su, Ali and Puller, Steven L.},
  title     = {Understanding Strategic Bidding in Multi-Unit Auctions: A Case Study of the {Texas} Electricity Spot Market},
  journal   = {The RAND Journal of Economics},
  year      = {2008},
  volume    = {39},
  number    = {1},
  pages     = {86--114},
  doi       = {10.1111/j.0741-6261.2008.00005.x},
}

@inproceedings{long2017peer,
  title={Peer-to-peer energy trading in a community microgrid},
  author={Long, Chao and Wu, Jianzhong and Zhang, Chenghua and Thomas, Lee and Cheng, Meng and Jenkins, Nick},
  booktitle={2017 IEEE power \& energy society general meeting},
  pages={1--5},
  year={2017},
  organization={IEEE}
}

@article{stigler1964theory,
  author  = {Stigler, George J.},
  title   = {A Theory of Oligopoly},
  journal = {Journal of Political Economy},
  year    = {1964},
  volume  = {72},
  number  = {1},
  pages   = {44--61}
}

@article{etesami2018stochastic,
  author    = {Etesami, S. Rasoul and Saad, Walid and Mandayam, Narayan B. and Poor, H. Vincent},
  title     = {Stochastic Games for the Smart Grid Energy Management with Prospect Prosumers},
  journal   = {IEEE Transactions on Automatic Control},
  year      = {2018},
  volume    = {63},
  number    = {8},
  pages     = {2327--2342},
  doi       = {10.1109/TAC.2018.2797217},
}

@article{zheng2022aieconomist,
  title={The {AI} Economist: Taxation policy design via two-level deep multiagent reinforcement learning},
  author={Zheng, Stephan and Trott, Alexander and Srinivasa, Sunil and Parkes, David C. and Socher, Richard},
  journal={Science Advances},
  volume={8},
  number={18},
  pages={eabk2607},
  year={2022}
}

@article{curry2022finding,
  title={Finding General Equilibria in Many-Agent Economic Simulations Using Deep Reinforcement Learning},
  author={Curry, Michael and Trott, Alexander and Phade, Soham and Bai, Yu and Zheng, Stephan},
  journal={arXiv preprint arXiv:2201.01163},
  year={2022}
}


\newpage
\appendix

\section{Proofs}\label{app:proofs}

\subsection{Proof of Proposition \ref{prop:private-existence}}
\begin{proof}
We verify that the required conditions R.1--R.5 of Myerson \& Reny (2020, Definition~9.1) are met. For this verification, we use an equivalent representation of the battery feasibility constraint that keeps each player's available action set independent of her signal (as required by R.1), while leaving the set of implemented battery flows (and thus the economic outcomes) unchanged.

Formally, at each date $t\in\{1,\dots,T\}$, let player $i$ choose an intended action $b_{it}\in[\underline b_i,\bar b_i]$
(independent of $S_{it}$). Given the current state of charge $S_{it}$, define the implemented battery flow
as the Euclidean projection onto the feasible interval,
\[
\tilde b_{it}\;:=\;\Pi_{B_i(S_{it})}(b_{it}) \;\in\; B_i(S_{it}),
\qquad\text{so that}\qquad
S_{i,t+1}\;=\;S_{it}+\tilde b_{it}.
\]
Imports/exports and payments are computed exactly as in Section~\ref{subsec:setup} but using the implemented flow $\tilde b_{it}$ in place of $b_{it}$ (so $\tilde x_{it}=\tilde b_{it}-\ell_i^t(\theta_i)$, $\tilde m_{it}=\tilde x_{it}^+$, $\tilde q_{it}=\tilde x_{it}^-$, and $\tilde c_{it}=p_c(\tilde M_t,\tilde Q_t)\tilde m_{it}-p_p(\tilde M_t,\tilde Q_t)\tilde q_{it}$). Because the projection satisfies $\Pi_{B_i(S_{it})}(b_{it})=b_{it}$ whenever $b_{it}\in B_i(S_{it})$, and because any $b_{it}\notin B_i(S_{it})$ induces the same implemented flow and next state as the boundary action $\tilde b_{it}\in B_i(S_{it})$, this reformulation does not change the set of attainable outcome paths under optimal play; it only embeds the feasibility constraint into the state transition. We can now verify the conditions R.1--R.5 for this equivalent representation in turn.

\emph{Condition R.1.}
The game has $T+1$ dates: at date $0$, nature draws the type profile $\theta\sim\pi$; at dates $t=1,\dots,T$, each player $i$ simultaneously chooses $b_{it}\in[\underline b_i,\bar b_i]$. Thus, taking $J=\{1\}$ in Myerson--Reny's notation, we may set $A_{it}=[\underline b_i,\bar b_i]$ for all signals $s_{it}$, so $\Phi_{it}(s_{it})=A_{it}$ for every $s_{it}$ as required by R.1. Nature has no moves at $t\ge 1$ (so $A_{0t}$ is a singleton for $t\ge 1$).

\emph{Condition R.2.}
Player $i$'s signal at $t=1$ is $s_{i1}=\theta_i$, a coordinate of nature's date-$0$ move. For $t\ge 2$, define
\[
s_{it}=(\theta_i,b_{i1},\dots,b_{i,t-1}),
\]
which is a literal projection of nature's date-$0$ coordinate and of player $i$'s own past action coordinates. In particular, $S_{it}$ is not itself a coordinate in the outcome history; it is computed from the signal via the recursion $S_{i1}=0$ and $S_{i,\tau+1}=S_{i\tau}+\Pi_{B_i(S_{i\tau})}(b_{i\tau})$ for $\tau<t$.

\emph{Conditions R.3--R.4.}
Each $A_{it}=[\underline b_i,\bar b_i]$ is a nonempty compact metric space and each $\Theta_i$ is finite (hence compact). Player $i$'s total payoff is $\sum_{t=1}^T u_{it}$ with $u_{it}=-\tilde c_{it}$, where $\tilde c_{it}$ is a continuous function of $(\theta,(b_{jt})_{j\in I},S_t)$ because: (i) projection onto an interval is continuous, (ii) the induced quantities $(\tilde m_{it},\tilde q_{it})$ are continuous in $\tilde b_{it}$, and (iii) $p_c$ and $p_p$ are continuous on the attainable set $K$ by Assumption~\ref{ass:prices1}. Thus $u_i$ is continuous in all actions and types, satisfying R.4.

\emph{Condition R.5.}
Nature moves only at date $0$, drawing $\theta$ from $\pi$ on the finite set $\Theta$. Finite-type priors satisfy the regularity requirement R.5 under the discrete topology: the Radon--Nikodym derivative is trivially continuous and its strict-positivity set is closed. At dates $t\ge 1$, nature's move set is a singleton.

Since all conditions R.1--R.5 are satisfied, Theorem~9.3 of Myerson \& Reny (2020) applies, and therefore for every $\varepsilon>0$ the private-information game possesses a perfect conditional $\varepsilon$-equilibrium.

In addition, we can also show that this result holds if, instead of excluding payments from signals, player $i$ observes her own period-$t$ bill. Then the above projected-signal verification can be maintained by adding a (purely formal) reporting stage: after players choose $(b_{jt})_{j\in I}$ at each date $t$, an additional player $0$ chooses a report vector $\hat c_t=(\hat c_{it})_{i\in I}$ in a compact interval $C:=\prod_{i\in I}[-\bar c,\bar c]$, where $\bar c<\infty$ bounds all attainable bills (the existence of such a bound follows directly from bounded flows, finite $\Theta$, and Assumption~\ref{ass:prices1}). Prosumers' period payoffs are $u_{it}=-\hat c_{it}$, and player $0$'s payoff is
\[
u_{0t}\;:=\;-\sum_{i\in I}\bigl(\hat c_{it}-\tilde c_{it}\bigr)^2,
\]
so player $0$ (approximately) reports the computed bill $\tilde c_{it}$. Define prosumer $i$'s signal to include her own past reports $(\hat c_{i1},\dots,\hat c_{i,t-1})$ in addition to $(\theta_i,b_{i1},\dots, \allowbreak b_{i,t-1})$. Because $\hat c_{it}$ is then a coordinate of player~$0$'s action history, the signal remains a literal projection and Conditions R.1--R.5 continue to hold (with player~$0$ added). Hence Theorem~9.3 yields existence of a perfect conditional $\varepsilon$-equilibrium for every $\varepsilon>0$ also in the variant with observed own bills.
\end{proof}

\subsection{Proof of Proposition \ref{prop:market-power}}
\begin{proof}
We prove the statements in turn.
\emph{Part (i)}. Fix a period $t$, a public state $(\theta, S_t)$, and continuous continuation values
$\{W_{i,t+1}(\theta,\cdot)\}_{i\in I}$. Each player's action set $B_i(S_{it})$ is a nonempty compact convex
interval, and player $i$'s payoff
\[
u_{it}(b_{it}, b_{-i,t}; \theta, S_t) \;+\; W_{i,t+1}(\theta, S_t + b_t)
\]
is continuous in all actions (Assumption~\ref{ass:prices1} and continuity of $W_{i,t+1}$).

Now suppose $W_{i,t+1}(\theta,\cdot)$ is concave in $S_{i,t+1}$ for each $i$. Fix $b_{-i,t}$. We show player $i$'s payoff is concave in $b_{it}$. First, on the export side ($b_{it}<\ell_i^t(\theta_i)$), let $q=\ell_i^t(\theta_i)-b_{it}>0$. Then the stage payoff term equals $q\,p_p(M_{-i,t},Q_{-i,t}+q)$, which is concave in $q$ by Assumption~\ref{ass:prices2}(ii), hence concave in $b_{it}$. Second, on the import side ($b_{it}>\ell_i^t(\theta_i)$), let $m=b_{it}-\ell_i^t(\theta_i)>0$. Then the stage payoff term equals $-m\,p_c(M_{-i,t}+m,Q_{-i,t})$, which is concave in $m$ by Assumption~\ref{ass:prices2}(i), hence concave in $b_{it}$. Third, at the kink $b_{it}=\ell_i^t(\theta_i)$, the left derivative with respect to $b_{it}$ is $-p_p(M_t,Q_t)$ and the right derivative is $-p_c(M_t,Q_t)$; since $p_p\le p_c$ (Assumption~\ref{ass:prices1}), the left derivative weakly exceeds the right derivative, so the stage payoff is globally concave in $b_{it}$. Finally, since $S_{i,t+1}=S_{it}+b_{it}$ is affine, concavity of $W_{i,t+1}(\theta,\cdot)$ in $S_{i,t+1}$ implies $b_{it}\mapsto W_{i,t+1}(\theta,S_t+b_t)$ is concave. Hence player $i$'s total payoff is concave in $b_{it}$.

Therefore, for each $i$, the best-response correspondence $BR_i(\cdot)$ is nonempty and compact-valued, and by concavity of the objective over a convex domain it is convex-valued (a weakly concave function on an interval has a convex argmax set). By Berge's maximum theorem, $\prod_{i\in I} BR_i(\cdot)$ is upper hemicontinuous. Hence the product correspondence $\prod_{i\in I} BR_i(\cdot)$ satisfies Kakutani's conditions, and the period-$t$ continuation game admits a pure-strategy Nash equilibrium. Moreover, the one-shot game is the special case $W_{i,t+1}\equiv 0$, so a pure-strategy Nash equilibrium exists.

To complete the backward-induction argument for the dynamic game (with a continuum of states), we select equilibria measurably at each date. Fix $t$ and continuation values $\{W_{i,t+1}(\theta,\cdot)\}_{i\in I}$. Let the public-state space be
\[
\Omega \;:=\; \Theta \times \prod_{j\in I}[0,\bar S_j],
\]
and let the action-product space be $A:=\prod_{j\in I}[\underline b_j,\bar b_j]$. For each $(\theta,S)\in\Omega$, write $B(\theta,S):=\prod_{j\in I}B_j(S_j)\subset A$, and define the set of pure Nash equilibria of the period-$t$ continuation game by the correspondence $\mathcal{E}_t:\Omega\rightrightarrows A$,
\[
\mathcal{E}_t(\theta,S)\;:=\;\Bigl\{b\in B(\theta,S):\ b_i\in BR_i(\theta,S,b_{-i})\ \ \forall i\in I\Bigr\}.
\]

We can then show that $\mathcal{E}_t$ admits a Borel measurable selection. Consider that we already showed $\mathcal{E}_t(\theta,S)\neq\emptyset$ for all $(\theta,S)$ by Kakutani; compactness then follows because $\mathcal{E}_t(\theta,S)\subset B(\theta,S)$ and $B(\theta,S)$ is compact.

Now we proceed to show that $\operatorname{Gr}(\mathcal{E}_t):=\{((\theta,S),b): b\in\mathcal{E}_t(\theta,S)\}$ is closed. Take any sequence $\{((\theta^n,S^n),b^n)\}_{n\ge 1}\subset \operatorname{Gr}(\mathcal{E}_t)$ with $((\theta^n,S^n),b^n)\to((\theta,S),b)$ in $\Omega\times A$. First, feasibility: for each $i$, $B_i(S_i)=[\max\{\underline b_i,-S_i\},\,\min\{\bar b_i,\bar S_i-S_i\}]$ has endpoints continuous in $S_i$,
hence $(\theta,S)\mapsto B(\theta,S)$ has closed graph, implying $b\in B(\theta,S)$. Second, best-response optimality: by Berge's maximum theorem (continuity of the objective and compact feasible sets), each $BR_i$ has closed graph. Because $b_i^n\in BR_i(\theta^n,S^n,b_{-i}^n)$ and $(\theta^n,S^n,b_{-i}^n,b_i^n)\to(\theta,S,b_{-i},b_i)$, we obtain $b_i\in BR_i(\theta,S,b_{-i})$ for every $i$. Hence $b\in\mathcal{E}_t(\theta,S)$, proving $\operatorname{Gr}(\mathcal{E}_t)$ is closed.

Then it follows that since $\Omega$ is a standard Borel space (finite $\Theta$ times a compact metric product) and $A$ is a compact metric space, a nonempty closed-valued correspondence with closed graph admits a Borel measurable selection (Kuratowski--Ryll-Nardzewski). 

We can then proceed by backward induction using such measurable selections. At $t=T$, $W_{i,T+1}\equiv 0$, so the period-$T$ game has a nonempty equilibrium correspondence $\mathcal{E}_T$ and we select a Borel measurable $b_T^*(\theta,S_T)\in\mathcal{E}_T(\theta,S_T)$. Define the (degenerate) Markov strategy $\beta_{iT}^*(\cdot\mid\theta,S_T)$ that assigns probability~$1$ to $b_{iT}^*(\theta,S_T)$, and define the induced value $W_{iT}^*(\theta,S_T)$ accordingly.

Now fix $t<T$ and assume measurable $\{b_{\tau}^*\}_{\tau=t+1}^T$ (equivalently $\{\beta_{\tau}^*\}$) and values $\{W_{i,t+1}^*\}_{i\in I}$ have been constructed. Applying the argument above to the period-$t$ continuation game with continuation $W_{i,t+1}^*$ yields a nonempty equilibrium correspondence $\mathcal{E}_t$, and we select a Borel measurable $b_t^*(\theta,S_t)\in\mathcal{E}_t(\theta,S_t)$. Let $\beta_t^*$ be the associated degenerate Markov strategy profile and define $W_{it}^*(\theta,S_t)$ by the recursive equation in the main text. Iterating for $t=T,T-1,\dots,1$ yields a pure-strategy Markov perfect equilibrium (MPE).

Finally, note that equilibrium existence in the dynamic observable-types game does not require the concavity condition above. By the same argument as in Proposition~\ref{prop:private-existence} (Myerson \& Reny, 2020, Theorem~9.3), for every $\varepsilon>0$ the observable-types dynamic game admits a perfect conditional $\varepsilon$-equilibrium. Formally, one may use the same intended-action/projection representation as in Proposition~\ref{prop:private-existence} to make players' available action sets independent of signals (Condition~R.1), and add a purely formal public reporting stage after each action profile that announces the public state (and, if included, realized payments) as an explicit coordinate of the outcome history so that signals are literal projections (Condition~R.2). The regularity conditions R.1--R.5 then continue to hold. The concavity condition is therefore used only to guarantee existence of a pure-strategy MPE.

\emph{Part (ii).}
Fix a period $t$, state $(\theta,S_t)$, and the associated one-shot game ($W_{i,t+1}\equiv 0$). Fix player $i$ and opponents' actions $b_{-i,t}$, and write
\[
q:=\ell_i^t(\theta_i)-b_{it}>0
\]
on the export branch, so that $M_t=M_{-i,t}$ and $Q_t=Q_{-i,t}+q$.

The strategic marginal payoff from increasing exports is
\[
D^\ast(q)=p_p(M_t,Q_t)+q\,\partial_Q p_p(M_t,Q_t).
\]
At an interior strategic export solution $q_{it}^\ast$, we have $D^\ast(q_{it}^\ast)=0$, hence
\[
p_p(M_t,Q_t)\big|_{q=q_{it}^\ast}
=
-\,q_{it}^\ast\,\partial_Q p_p(M_t,Q_t)\big|_{q=q_{it}^\ast}\ge 0,
\]
with strict inequality whenever $\partial_Q p_p<0$ at the relevant arguments. 
By the definition of the price-taking benchmark, the price-taking marginal payoff on the export branch is
\[
D^{PT}(q)=p_p(M_t,Q_t),
\]
which is weakly decreasing in $q$ by Assumption~\ref{ass:prices2}(ii). Since $D^{PT}(q_{it}^\ast)\ge 0$, every price-taking optimizer satisfies
\[
q_{it}^{PT}\ge q_{it}^\ast,
\]
with strict inequality if $\partial_Q p_p<0$ at the relevant arguments.

For the social planner comparison, holding $b_{-i,t}$ fixed, the one-shot planner chooses $b_{it}$ to maximize $-C_t(Y_t)$. On the export branch, increasing $q$ lowers $Y_t$ one-for-one, so the planner's marginal payoff is
\[
C_t'(Y_t)\in [P^E,P^I]\subseteq \mathbb R_+.
\]
Hence the planner chooses the maximal feasible export on the export branch. The same is true for the price-taking benchmark, since its marginal payoff is
\[
D^{PT}(q)=p_p(M_t,Q_t)\ge P^E\ge 0
\]
by Assumption~\ref{ass:prices1}. Therefore (up to tie-breaking),
\[
q_{it}^{PT}=q_{it}^{SP}.
\]
Combining the two comparisons yields
\[
q_{it}^{\ast}\le q_{it}^{PT}=q_{it}^{SP},
\]
with strict inequality in the first comparison under the stated strictness condition.

\emph{Part (iii).}
Fix $t<T$, player $i$, state $(\theta,S_t)$, and a continuation path $(b_{-i,\tau})_{\tau=t}^T$. Let $W_{i,\tau}^{\ast}$ and $W_{i,\tau}^{PT}$ denote the continuation values induced by the strategic and price-taking problems against this fixed path, and define the marginal continuation values
\[
\lambda_{i,\tau}^{k}(s):=\partial_s W_{i,\tau}^{k}(s),\qquad k\in\{\ast,PT\},
\]
on the relevant continuation-state interval. Assume that the marginal unit of storage at any future date $\tau\ge t+1$ is either carried forward or exported. Hence, for each $\tau$ and relevant $s$, the shadow value satisfies
\[
\lambda_{i,\tau}^{k}(s)=\max\{\rho_{i,\tau}^{k}(s),\,\lambda_{i,\tau+1}^{k}(s)\},
\qquad \lambda_{i,T+1}^{k}\equiv 0,
\]
where $\rho_{i,\tau}^{k}(s)$ is the marginal payoff from exporting the marginal unit at date $\tau$.
By the definition of the price-taking benchmark, on the export branch
\[
\rho_{i,\tau}^{\ast}(s)\le \rho_{i,\tau}^{PT}(s)
\]
because $\partial_Q p_p\le 0$ (Assumption~\ref{ass:prices2}(ii)). Backward induction on $\tau$ then yields
\[
\lambda_{i,t+1}^{\ast}(s)\le \lambda_{i,t+1}^{PT}(s)
\]
on the relevant interval.

Now consider period $t$ on the import branch and write $s=S_{it}+b_{it}$. The strategic and price-taking marginal payoffs from increasing $b_{it}$ are
\begin{align*}
D^\ast(b_{it})
&=
-p_c(M_t,Q_t)-m_{it}(b_{it})\,\partial_M p_c(M_t,Q_t)+\lambda_{i,t+1}^{\ast}(s),\\
D^{PT}(b_{it})
&=
-p_c(M_t,Q_t)+\lambda_{i,t+1}^{PT}(s).
\end{align*}
Therefore,
\[
D^{PT}(b_{it})-D^\ast(b_{it})
=
m_{it}(b_{it})\,\partial_M p_c(M_t,Q_t)
+\bigl[\lambda_{i,t+1}^{PT}(s)-\lambda_{i,t+1}^{\ast}(s)\bigr]\ge 0,
\]
since $m_{it}(b_{it})>0$ on the import branch, $\partial_M p_c\ge 0$ (Assumption~\ref{ass:prices2}(i)), and $\lambda_{i,t+1}^{PT}\ge \lambda_{i,t+1}^{\ast}$.

By concavity of the period-$t$ objectives, $D^\ast$ and $D^{PT}$ are decreasing and the interior optimizers satisfy
\[
D^\ast(b_{it}^{\ast})=0,\qquad D^{PT}(b_{it}^{PT})=0.
\]
Hence
\[
D^{PT}(b_{it}^{\ast})\ge D^\ast(b_{it}^{\ast})=0,
\]
implying
\[
b_{it}^{\ast}\le b_{it}^{PT}.
\]

\emph{Part (iv).}
Let
\[
\Delta:=\sum_i \bigl(b_{it}^{SP,t}-b_{it}^{PT}\bigr)>0.
\]
By assumption, this additional period-$t$ charge is fully discharged through future exports. Therefore the cumulative future export difference is strictly positive:
\[
\sum_{\tau=t+1}^T \bigl(Q_\tau^{SP,t}-Q_\tau^{PT}\bigr) > 0.
\]
Hence not all summands can be weakly negative, and there exists some $\tau>t$ such that
\[
Q_\tau^{SP,t}>Q_\tau^{PT}.
\]
\end{proof}

\section{Computational details}\label{app:computational}

This appendix provides full details of the neural-network architectures, training procedure, and experimental configurations used in Section~\ref{sec:results}.

\subsection{Architecture}\label{app:architecture}

Each controllable agent~$i$ is assigned an independent feedforward neural network $f^{\phi^i}$ with the architecture shown in Table~\ref{tab:arch-dec}. The network maps the agent's observation vector to a scalar battery action $a_{it} \in [-1, 1]$, which is scaled to the feasible battery flow $b_{it} = a_{it} \cdot \bar b_i$ and clamped to the capacity-feasible set $B_i(S_{it})$. PV curtailment is disabled (all available PV is used). All weights are initialised with orthogonal initialisation (gain~$0.5$) and biases are set to zero.

\begin{table}[h!]
\centering
\caption{Decentralised policy network architecture (per agent).}
\label{tab:arch-dec}
\begin{tabular}{@{}lll@{}}
\toprule
Layer & Specification & Output size \\
\midrule
Input   & Observation vector & $n_{\mathrm{in}} \in \{5, 9\}$  \\
Hidden~1 & Linear($n_{\mathrm{in}}$, 32) + ReLU & 32 \\
Hidden~2 & Linear(32, 32) + ReLU & 32 \\
Output  & Linear(32, 1) + Tanh & 1 \\
\bottomrule
\end{tabular}
\end{table}

\noindent
The observation vector depends on the information-structure treatment. Under \emph{private} information ($n_{\mathrm{in}}=5$), the agent observes her normalised battery state $S_{it}/\bar S_i$, the time of day $t/T$, normalised demand, normalised PV generation, and the normalised midpoint price $(p_p + p_c)/(2 P^I)$. Under \emph{price-observable} information ($n_{\mathrm{in}}=9$), the agent additionally observes the normalised sell price, buy price, bid--ask spread, average total demand, and average total supply. When endogenous tatonnement is active, the price features are derived from the current tatonnement price estimate $\hat p_t^{(k)}$ rather than from lagged prices.

The social planner uses a single centralised network that observes the full state and outputs joint actions for all agents. The architecture mirrors the decentralised network but with input size $1 + 3N$ (time of day, plus normalised battery state, demand, and PV for each of the $N$ agents) and output size $N$ (one action per agent; passive agents' outputs have no effect). The hidden-layer width is $16N$, so the network capacity scales with the number of agents. Standard PyTorch initialisation is used (no orthogonal).

\subsection{Training procedure}\label{app:training}

At each period~$t$, agents' actions and market prices are determined jointly through an endogenous tatonnement procedure. Starting from an initial price estimate $\hat p_t^{(0)}$ (the previous period's realised price, or the tariff midpoint at $t=1$), the procedure iterates: (i) each agent computes her best-response action $b_{it}^{(k)}$ given $\hat p_t^{(k)}$; (ii) the implied aggregate quantities $(M_t^{(k)}, Q_t^{(k)})$ and new prices $p_t^{(k)}$ are computed; (iii) the price estimate is updated with damping, $\hat p_t^{(k+1)} = \alpha\, p_t^{(k)} + (1-\alpha)\,\hat p_t^{(k)}$. The iteration terminates after at most $N_K$ steps or when the price change falls below the tolerance. In strategic mode, no gradients are detached during the iteration: the entire fixed-point chain carries gradients, so agents' policy updates account for the equilibrium relationship between their actions and the resulting prices. In price-taking mode, the iterations are performed without gradients and only the final forward pass carries gradients through each agent's own net import to her cost.

Table~\ref{tab:hyperparams} lists all training hyperparameters, used consistently across all treatments unless otherwise noted.

\begin{table}[h!]
\centering
\caption{Training hyperparameters.}
\label{tab:hyperparams}
\begin{tabular}{@{}llr@{}}
\toprule
Parameter & Symbol & Value \\
\midrule
Training episodes & & 350 \\
Mini-batches per episode & & 5 \\
Mini-batch size & $N_B$ & 32 \\
Training days & & 1{,}000 \\
Test days & & 100 \\
Optimiser & & Adam \\
Learning rate & $\eta$ & $5 \times 10^{-3}$ \\
Gradient clipping & & max norm $1.0$ \\
\addlinespace
Initial noise scale & $\sigma_0$ & 0.5 \\
Noise type & & Additive Gaussian \\
Noise decay & & Linear: $\sigma_0 (1 - e/300)$ \\
\addlinespace
Tatonnement iterations & $K$ & 6 \\
Tatonnement damping & $\alpha$ & 0.7 \\
Tatonnement tolerance & & $10^{-5}$ \\
Tatonnement mode (strategic) & & Endogenous (full gradient) \\
Tatonnement mode (price-taking) & & Detached (no gradient) \\
\addlinespace
Regret episodes & & 350 \\
Random seeds & & $\{42, 43, 44, 45, 46\}$ \\
\bottomrule
\end{tabular}
\end{table}

For the decentralised networks, exploration noise is additive Gaussian: at each period, the network's output action is perturbed as $a_{it} \leftarrow \mathrm{clamp}(a_{it} + \varepsilon, -1, 1)$ with $\varepsilon \sim \mathcal{N}(0, \sigma_e^2)$ and $\sigma_e = \sigma_0 \cdot \max(1 - e/300,\, 0)$. Noise is applied only on the final tatonnement iteration to avoid contaminating the price-discovery process. For the centralised planner, exploration uses random replacement: each action is replaced with a uniform draw from $[-1, 1]$ with probability $\sigma_e = \sigma_0 \cdot 0.9^e$, decaying exponentially.

In strategic mode, per-agent gradients are computed via \texttt{torch.autograd.grad}$(\bar C^i, \phi^i)$ for each controllable agent~$i$, isolating the gradient $\nabla_{\phi^i} \bar C^i$ from cross-agent terms so that each agent's update reflects her own best response (see Section~\ref{subsec:madh}). In price-taking mode, aggregate quantities are detached from the computation graph before price calculation, making sequential \texttt{backward()} calls safe and equivalent. In both modes, gradients are clipped to unit norm before the optimiser step. During training, initial battery states are drawn uniformly from $[0, \bar S_i]$ in every episode; during evaluation, all batteries are initialised at $S_{i1} = 0$ to match the model specification.

\subsection{Experimental configurations}\label{app:experiments}

Table~\ref{tab:experiments} summarises the three treatments. Each treatment trains all relevant benchmarks over five seeds. Abbreviations: SP = social planner, PT = price-taking, S = strategic (Nash).

\begin{table}[h!]
\footnotesize
\centering
\caption{Summary of experimental treatments.}
\label{tab:experiments}
\begin{tabular}{@{}llllll@{}}
\toprule
Section & Treatment & Community & Mechanism & Info & Benchmarks \\
\midrule
\ref{subsec:baseline}      & Baseline ordering           & Standard (A--F)        & SDR              & Private & SP, PT, S \\
\ref{subsec:competition}   & Competition                 & $N{\times}$C + D + E   & SDR              & Private & PT, S \\
\ref{subsec:heterogeneity} & Heterogeneity \& mechanisms & Standard (A--F)        & SDR, MMR, LIN    & Private & PT, S \\
\bottomrule
\end{tabular}
\end{table}

\noindent
In the competition treatment, the number of symmetric storage agents varies over $N \in \{1, 2, 3\}$. Total battery capacity is fixed at $10\,$kWh and divided equally, so each agent receives $\bar S = 10/N\,$kWh with charge/discharge rate $\bar b = 10/(3N)\,$kW; two passive agents (types~D and~E) provide background demand and supply. In the information-design treatment, the two information structures differ only in the policy network's input vector (5 vs.\ 9 features, as described above); all other hyperparameters are identical. In the scaling treatment, the standard six-agent community is replicated $k \in \{1, 2, 3\}$ times, yielding communities of $6$, $12$, and $18$ agents ($3k$ controllable). The decentralised hidden-layer width remains at $32$ for all agents; the centralised hidden-layer width scales as $16 \times 6k$ (the total number of agents).

\subsection{Stochastic data generation}\label{app:data}

Each agent's PV generation is $s_i^t = \bar s_i \cdot \sin \bigl(\pi(t - 6)/(18 - 6)\bigr)\,\mathbf{1}_{[6,18]}(t)$, perturbed by a mean-reverting AR(1) scaling process $z_t^{\mathrm{PV}} = \mu + \varphi (z_{t-1} - \mu) + \varepsilon_t$ with $\varphi=0.85$, $\sigma_\varepsilon=0.20$, and daily mean $\mu\sim U(0.6, 0.95)$. Additionally, $1$--$4$ Gaussian cloud-dip transients (depth $20$--$70\%$, width $1$--$2\,$h) are superimposed. The scaling factor is clipped to $[0.02, 1.0]$ and drawn per day (shared across agents, with agent-specific $\bar s_i$).

Demand follows a composite time-of-use profile (morning ramp 5--8h, low midday 9--17h, evening peak 17--22h), perturbed by an AR(1) process ($\varphi^d = 0.6$, $\sigma^d = 0.25$) with $1$--$3$ random activity pulses per day. Daily total demand is scaled by a log-normal factor ($\sigma = 0.2$). Average daily demand varies by agent type (Table~\ref{tab:agents}).

\subsection{Convergence}\label{app:convergence}

Figure~\ref{fig:learning-curves} shows training convergence for the baseline community (SDR, private information). Both the strategic and price-taking learners converge smoothly within approximately 100 episodes. The strategic cost lies persistently above the price-taking cost throughout training, reflecting the prisoners' dilemma: even during learning, gradient-based optimisers that internalise the price impact converge to a costlier equilibrium than those that ignore it. The gap between the two curves stabilises early and remains essentially constant after convergence, indicating that the welfare loss from market power is a robust equilibrium property rather than a transient artefact of the training process.

\begin{figure}[h!]
\centering
\includegraphics[width=0.48\textwidth]{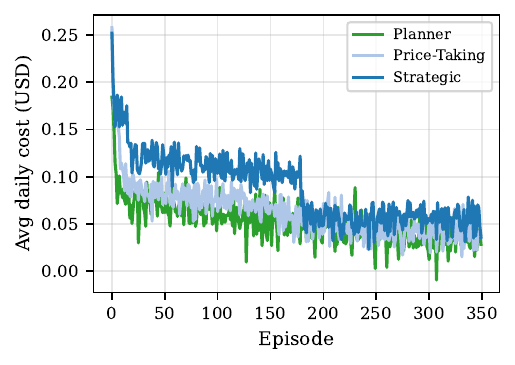}
\caption{Training convergence (SDR, private information). Average daily grid cost per episode for strategic and price-taking learning. The strategic cost converges to a persistently higher level, reflecting the cost of internalising price impact.}
\label{fig:learning-curves}
\end{figure}

\section{Verification of Assumption~\ref{ass:prices2} and concave continuation values}\label{app:verification}

This appendix verifies, for each price mechanism, (a) whether Assumption~\ref{ass:prices2} holds and (b) the concavity condition on equilibrium continuation values required by Proposition~\ref{prop:market-power}(i) for pure-strategy MPE existence. Recall from the proof of Proposition~\ref{prop:market-power}(i) that we require each player's total period-$t$ objective---the sum of the stage payoff and the continuation value---to be concave in her own battery action $b_{it}$. On the export branch ($q>0$), this requires concavity of the export-revenue function $q\,p_p(M,Q_{-i}+q)$ in $q$; on the import branch ($m>0$), this requires convexity of the import-cost function $m\,p_c(M_{-i}+m,Q)$ in $m$; and across both branches, the continuation value $W_{i,t+1}(\theta,\cdot)$ must be concave in $S_{i,t+1}$.

Throughout this appendix, write $M=M_{-i}+m$ for total imports and $Q=Q_{-i}+q$ for total exports when player~$i$'s quantity varies. For all mechanisms, $P^E=0.05$, $P^I=0.30$, and $p_{\mathrm{mid}}=(P^E+P^I)/2=0.175$. We use the general identity
\begin{equation}\label{eq:d2-general}
\frac{d^2}{dm^2}\bigl[m\,p_c(M_{-i}+m,Q)\bigr]
= 2\,\frac{\partial p_c}{\partial M}(M,Q)
+ m\,\frac{\partial^2 p_c}{\partial M^2}(M,Q),
\end{equation}
and the analogous expression for $d^2[q\,p_p]/dq^2$ with $Q=Q_{-i}+q$.

\subsection{Linear mechanism}\label{app:linear-verification}

The linear mechanism admits a fully analytical verification.
In the $M\ge Q$ regime (demand exceeds supply), $p_p = p_{\mathrm{mid}}+\kappa(M-Q)$ and the revenue is
\[
R(q) = q\,\bigl(p_{\mathrm{mid}}+\kappa M - \kappa Q_{-i} - \kappa q\bigr),
\]
so $d^2 R/dq^2 = -2\kappa < 0$ for all parameter values. In the $M<Q$ regime, $p_p = P^E - \kappa M + MD_\kappa/Q$ where $D_\kappa=(P^I-P^E)/2+\kappa M>0$, and applying~\eqref{eq:d2-general} yields
\[
\frac{d^2 R}{dq^2} = -\frac{2M\,D_\kappa\,Q_{-i}}{Q^3} < 0 \qquad\text{whenever }Q_{-i}>0.
\]
At the regime boundary $Q=M$, the slope $dR/dq$ drops discontinuously (from $p_{\mathrm{mid}}-\kappa(Q_{-i}-M)+\kappa M$ to a lower value), so the kink preserves global concavity. When $Q_{-i}=0$, $d^2R/dq^2=0$ in the $M<Q$ regime (linear). Hence $R(q)$ is globally weakly concave, and strictly concave whenever $Q_{-i}>0$ or $M\ge Q$.

In the $M<Q$ regime, $p_c = p_{\mathrm{mid}}-\kappa Q+\kappa(M_{-i}+m)$ is affine in $m$, so $C(m)=m\,p_c$ is quadratic with $d^2 C/dm^2 = 2\kappa > 0$ for all parameter values. In the $M\ge Q$ regime, $p_c = P^I + \kappa Q + AQ/M$ where $A = p_{\mathrm{mid}}-\kappa Q - P^I < 0$, and~\eqref{eq:d2-general} yields
\[
\frac{d^2 C}{dm^2} = \frac{-2AQ\,M_{-i}}{M^3} > 0 \qquad\text{whenever }M_{-i}>0,
\]
since $A<0$. At the regime boundary $M=Q$, the slope $dC/dm$ jumps upward (from $p_{\mathrm{mid}}$ to $p_{\mathrm{mid}}+m(P^I-p_{\mathrm{mid}})/Q > p_{\mathrm{mid}}$), preserving global convexity. When $M_{-i}=0$, $d^2C/dm^2=0$ in the $M\ge Q$ regime. Hence $C(m)$ is globally weakly convex, and strictly convex whenever $M_{-i}>0$ or $M<Q$.

Within each smooth regime, the curvature conditions hold. The stage payoff is concave in $b_{it}$ on both the export and import branches (strictly so whenever $Q_{-i}>0$ or $M_{-i}>0$, respectively). Since the within-day stage payoff is quadratic on each branch, the sum of a (weakly) concave stage payoff and a concave continuation value is concave, so the Kakutani argument in the proof of Proposition~\ref{prop:market-power}(i) applies. However, the analysis above applies only within each smooth pricing regime. The linear mechanism clamps the clearing price $p^{\mathrm{LIN}}$ to the corridor $[P^E, P^I]$. When the clearing price hits the upper bound $P^I$ (at $M-Q = (P^I - p_{\mathrm{mid}})/\kappa$), the consumer price $p_c$ becomes constant at $P^I$ for larger~$M$, and the import-cost slope $dC/dm$ drops from a value above~$P^I$ to exactly~$P^I$. This creates a concave kink that breaks global convexity of the import-cost function. Symmetrically, when the clearing price hits $P^E$, a convex kink in the export-revenue function breaks global concavity.

Within each smooth regime (between the kinks at $M=Q$ and at the clamping boundaries), the curvature conditions hold as derived above. The regime-boundary kink at $M=Q$ preserves convexity/concavity (as shown), but the clamping kinks do not.

\subsection{MMR mechanism}\label{app:mmr-verification}

In the $M\ge Q$ regime, $p_p = p_{\mathrm{mid}}$ is constant, so $R(q)=q\,p_{\mathrm{mid}}$ is linear ($d^2R/dq^2=0$). In the $M<Q$ regime, $p_p = P^E + M(p_{\mathrm{mid}}-P^E)/Q$, and~\eqref{eq:d2-general} yields
\[
\frac{d^2 R}{dq^2} = -\frac{2M\,Q_{-i}\,(p_{\mathrm{mid}}-P^E)}{Q^3} < 0 \qquad\text{whenever }Q_{-i}>0.
\]
At the regime boundary $Q=M$, the slope drops from $p_{\mathrm{mid}}$ (left) to $p_{\mathrm{mid}} - q(p_{\mathrm{mid}}-P^E)/M < p_{\mathrm{mid}}$ (right), preserving global concavity. Hence $R(q)$ is globally weakly concave, strictly concave in the $M<Q$ regime whenever $Q_{-i}>0$.

In the $M<Q$ regime (supply exceeds demand), $p_c=p_{\mathrm{mid}}$ is constant, so $C(m)=m\,p_{\mathrm{mid}}$ is linear ($d^2C/dm^2=0$). In the $M\ge Q$ regime, $p_c = P^I - Q(P^I-p_{\mathrm{mid}})/M$, and~\eqref{eq:d2-general} yields
\[
\frac{d^2 C}{dm^2} = \frac{2Q\,M_{-i}\,(P^I-p_{\mathrm{mid}})}{M^3} > 0 \qquad\text{whenever }M_{-i}>0.
\]
At the regime boundary $M=Q$, the slope jumps from $p_{\mathrm{mid}}$ (left) to $p_{\mathrm{mid}} + m(P^I-p_{\mathrm{mid}})/Q > p_{\mathrm{mid}}$ (right), preserving global convexity. Hence $C(m)$ is globally weakly convex, strictly convex in the $M\ge Q$ regime whenever $M_{-i}>0$.

\paragraph{Conclusion for the MMR mechanism.}
Assumption~\ref{ass:prices2} holds globally with weak curvature. In the regime where a player has price impact (the ``long side''), the curvature is strict. In the constant-price regime (the ``short side''), the cost or revenue is linear; however, the continuation value $W_{i,t+1}$ provides the concavity needed for a unique best response in the dynamic game.

\subsection{SDR mechanism}\label{app:sdr-verification}

In the $Q\le M$ regime ($\rho\le 1$), $p_p = P^E P^I M / D$ where $D=(P^I-P^E)Q + P^E M$. A direct computation yields
\[
\frac{d^2 R}{dq^2} = -\frac{2\,K_1\,(P^I-P^E)}{D^3},
\qquad K_1 = P^E P^I M\bigl[(P^I-P^E)Q_{-i}+P^E M\bigr] > 0.
\]
Since $P^I>P^E$ and $D>0$, we have $d^2R/dq^2<0$ for all attainable $(M,Q)$ with $Q\le M$. In the $Q>M$ regime ($\rho>1$), $p_p=P^E$ and $R(q)=q\,P^E$ is linear. In the $Q>M$ regime ($\rho>1$), $p_p=P^E$ and $R(q)=q\,P^E$ is linear with slope~$P^E$. At the boundary $Q=M$, the left-hand slope (from the $\rho\le 1$ regime) equals
\[
P^E + (M-Q_{-i})\,\frac{P^E(P^E-P^I)}{P^I M} \;<\; P^E,
\]
because $\partial_Q p_p < 0$ pulls the marginal revenue below the price. The slope therefore \emph{increases} at the boundary, creating a convex kink. Consequently, the export-revenue function is \emph{not} globally concave: it is strictly concave within the $\rho\le 1$ regime and linear in the $\rho>1$ regime, but the kink at $Q=M$ violates global concavity.

In the $Q>M$ regime ($\rho>1$), $p_c = P^E$ and $C(m)=m\,P^E$ is linear. In the $Q\le M$ regime, $p_c$ is the sum of two terms with competing curvature:
\[
p_c(M,Q) = \underbrace{\frac{P^E P^I Q}{D}}_{\text{convex in }M} + \underbrace{P^I \left(1-\frac{Q}{M}\right)}_{\text{concave in }M},
\]
where $D=(P^I-P^E)Q+P^E M$. Applying~\eqref{eq:d2-general},
\[
\frac{d^2 C}{dm^2}
= 2 \left[\frac{P^I Q}{M^2}-\frac{(P^E)^2 P^I Q}{D^2}\right]
+ m \left[\frac{2(P^E)^3 P^I Q}{D^3}-\frac{2 P^I Q}{M^3}\right].
\]
For moderate $m$ (i.e.\ $\rho$ bounded away from zero), the first bracketed term dominates and $d^2C/dm^2>0$. However, as $m\to\infty$ with $Q$ and $M_{-i}$ fixed, $D\approx P^E M$ and both terms in the second bracket converge to $2P^I Q/M^3$ with opposite signs, while the first bracket converges to zero. Numerically, $d^2C/dm^2$ changes sign from positive to negative at a threshold that depends on the ratio $m/Q$. For example, with $M_{-i}=1$ and $Q=3$, the sign change occurs at $m\approx 8.9$ (i.e.\ $\rho\approx 0.30$); with $M_{-i}=0.5$ and $Q=0.5$, it occurs at $m\approx 2.6$ ($\rho\approx 0.16$).

As $\rho\to 0$, $p_c\to P^I$ and the cost function $C(m) = m\,p_c$ approaches $m\,P^I$ (linear). The hyperbolic structure of the SDR formula causes $C(m)$ to approach this linear asymptote from above, creating a region of slight concavity. This means Assumption~\ref{ass:prices2}(i) \emph{fails} for the SDR mechanism at extreme quantity ratios, even in its weak (non-strict) form.

On the import branch at period $T$ (where $W_{i,T+1}\equiv 0$), the stage payoff $-C(m)$ is not globally concave in $m$, so the standard one-shot Kakutani argument does not apply directly. Two observations mitigate this. First, Proposition~\ref{prop:private-existence} guarantees existence of a perfect conditional $\varepsilon$-equilibrium for every $\varepsilon>0$ via the \cite{myerson2020epsilon} argument, which requires no curvature assumptions. Second, for the dynamic game ($t<T$), the total objective includes the continuation value, which is concave in $S_{i,t+1}$ and hence in $b_{it}$; this additional concavity can dominate the slight non-convexity of the stage cost in the region where the violation occurs. We verify this numerically below.

\subsection{Verification of concave equilibrium continuation values}\label{app:verification}

We additionally numerically verifies, for each price schedule mechanism considered in the paper, the concavity-propagation condition required by Proposition~\ref{prop:market-power} namely that the equilibrium value function $W_{it}^\ast(\theta,S_t)$ inherits concavity in $S_{it}$ from a concave continuation $W_{i,t+1}^\ast(\theta,S_{t+1})$. We first explain the approach and then report the results of the individual calculations for each mechanism. Note that for the linear pricing mechanism, the result can be derived analytically and the verification focuses on SDR and MMR.

\subsection{Numerical verification approach}
For each mechanism, the numerical verification proceeds by backward induction on a finite grid for the public battery state. Fix a type profile $\theta$ and hence the induced period surpluses $\ell_i^t(\theta_i)$. In each period $t$, we discretize $(S_{it},S_{-i,t})$ on a uniform grid over the feasible state space and, at each grid point, solve the one-period continuation game in battery flows $(b_{it},b_{-i,t})$ given the continuation value from period $t+1$. This yields equilibrium actions and the induced equilibrium value $W_{it}^\ast(\theta,S_t)$ on the state grid. We then test concavity in the own battery state $S_{it}$ by computing discrete second differences of the grid values while holding $S_{-i,t}$ fixed (slice-by-slice), and repeat this recursively from $t=T$ back to $t=1$. Thus, the verification checks directly whether concavity is preserved by the equilibrium Bellman operator for the mechanism under consideration.

To make the verification feasible on finer grids, we use analytical derivatives of the pricing mechanisms to solve the repeated best-response optimization at each state grid point. In particular, for each mechanism we calculate $\partial_M p_c(M,Q)$ and $\partial_Q p_p(M,Q)$ (and the corresponding second derivatives on smooth branches), which yield analytical derivatives of the stage payoff in the importer and exporter cases. Combined with the derivative of the continuation term $W_{i,t+1}^\ast(\theta,S_{t+1})$ with respect to $b_{it}$, this converts the best-response problem on each smooth segment into a root-finding problem for the first-order condition. The solver partitions the feasible interval at all relevant kinks---(i) the trade-regime kink at $b_{it}=\ell_i^t(\theta_i)$, (ii) mechanism-specific pricing kinks such as the SDR/MMR branch boundaries, and (iii) continuation interpolation kinks---and applies a root solver on each smooth subinterval. If needed, we additionally fall back to a bounded scalar minimization routine. The use of the derivatives of the price mechanisms reduces computation time substantially, but note that the concavity check itself is performed on the computed grid values.

We report multiple verifications for each mechanism. These vary in particular by the horizon length $T$, the state-grid size $n_{\text{grid}}$, and the surplus path $(\ell_i^t,\ell_{-i}^t)_{t=1}^T$.\footnote{We keep solver parameters regarding the best-response tolerance, maximum iterations, and damping constant across runs.} We simplify the surplus path variation by distinguishing three types: whether players' realizations are symmetric, asymmetric, or a mix across periods. For each run, we report as ``Result'' whether the concavity test passes for the entire run (all periods, both players, and all own-state slices) at the specified numerical tolerance. ``Worst $\Delta^2$'' reports the largest discrete second difference observed in the concavity checks across all reported slices and periods; positive values indicate local convexity on the grid, while very small positive values (near the tolerance) are typically numerical error rather than economically meaningful violations. 
Finally, ``Time (s)'' is the wall-clock runtime of the scenario.

\begin{table}[htbp]
\centering
\small
\caption{Initial numerical concavity-verification runs for SDR and MMR}
\label{tab:verification-runs}
\begin{tabular}{llrrlrrr}
\toprule
Mechanism & Type & $T$ & $n_{\text{grid}}$ & Interp. & Result & Worst $\Delta^2$ & Time (s) \\
\midrule
SDR & Sym  & 2 & 25 & linear & PASS    & 1.03e-09 & 27.3   \\
SDR & Sym  & 2 & 50 & linear & PASS    & 2.12e-09 & 212.8  \\
SDR & Asym & 2 & 25 & linear & PASS    & 2.09e-09 & 552.7  \\
SDR & Mix  & 2 & 25 & linear & PASS    & 1.95e-09 & 522.0  \\
\addlinespace
SDR & Sym  & 4  & 25  & linear & PENDING & -- & -- \\
SDR & Asym & 4  & 25  & linear & PENDING & -- & -- \\
SDR & Mix  & 4  & 25  & linear & PENDING & -- & -- \\
\addlinespace
MMR & Sym  & 2  & 25  & linear & PENDING & -- & -- \\
MMR & Sym  & 2  & 50  & linear & PENDING & -- & -- \\
MMR & Asym & 2  & 25  & linear & PENDING & -- & -- \\
MMR & Mix  & 2  & 25  & linear & PENDING & -- & -- \\
\bottomrule
\end{tabular}
\vspace{0.35em}
\parbox{0.96\linewidth}{\footnotesize
\emph{Notes:} ``Type'' denotes the surplus-path structure (Sym = symmetric, Asym = asymmetric, Mix = mixed across periods). $T$ is the horizon length and $n_{\text{grid}}$ is the number of grid points per state dimension. ``Result'' reports whether the discrete concavity test passes for the full run. ``Worst $\Delta^2$'' is the largest discrete second difference observed across all slices and periods; positive values indicate local convexity on the grid. ``PENDING'' indicates runs deferred to the post-submission version. For MMR, stage-payoff curvature is established analytically in Appendix~\ref{app:mmr-verification}; the pending entries concern the continuation-value verification only.
}
\end{table}

\end{document}